%% file: main.tex
\documentclass[11pt]{article}

\usepackage[utf8]{inputenc}
\usepackage{float}
\usepackage{algorithm}
\usepackage{algpseudocode}

\usepackage{amsmath,amsfonts,amsthm,amssymb,color}
\usepackage[dvipsnames,svgnames,table]{xcolor}
\definecolor{darkgreen}{rgb}{0.0,0,0.9}
\usepackage{mathtools}
\usepackage{authblk}
\usepackage{fullpage}
\usepackage{todonotes}
\usepackage{parskip}
\usepackage{comment}
\usepackage{tikz}
\usepackage{bbm}
\usepackage{dsfont}
\usepackage[sc]{mathpazo}
\usepackage[basic]{complexity}
\usepackage[colorlinks=true,citecolor=OliveGreen,linkcolor=BrickRed,urlcolor=BrickRed,pdfstartview=FitH]{hyperref}
\usepackage[capitalize,nameinlink]{cleveref}
\usepackage{tcolorbox}
\usepackage[short]{optidef}

\newtcolorbox{wbox}
{
	colback  = white,
}

\usepackage{amsthm}
\usepackage{aliascnt}
\usepackage[colorlinks=true,citecolor=OliveGreen,linkcolor=BrickRed,urlcolor=BrickRed,pdfstartview=FitH]{hyperref}
\usepackage[capitalize,nameinlink]{cleveref}

\theoremstyle{plain}
\newtheorem{theorem}{Theorem}[section]
\crefname{theorem}{theorem}{theorems}
\Crefname{theorem}{Theorem}{Theorems}

\newtheorem{lemma}{Lemma}[section]
\crefname{lemma}{lemma}{lemmas}
\Crefname{lemma}{Lemma}{Lemmas}

\newaliascnt{corollary}{lemma}
\newtheorem{corollary}[corollary]{Corollary}
\aliascntresetthe{corollary}
\crefname{corollary}{corollary}{corollaries}
\Crefname{corollary}{Corollary}{Corollaries}

\newaliascnt{proposition}{lemma}
\newtheorem{proposition}[proposition]{Proposition}
\aliascntresetthe{proposition}
\crefname{proposition}{proposition}{propositions}
\Crefname{proposition}{Proposition}{Propositions}

\newaliascnt{claim}{lemma}

\aliascntresetthe{claim}
\crefname{claim}{claim}{claims}
\Crefname{claim}{Claim}{Claims}

\newaliascnt{fact}{lemma}

\aliascntresetthe{fact}
\crefname{fact}{fact}{facts}
\Crefname{fact}{Fact}{Facts}

\theoremstyle{definition}

\newaliascnt{definition}{lemma}
\newtheorem{definition}[definition]{Definition}
\aliascntresetthe{definition}
\crefname{definition}{definition}{definitions}
\Crefname{definition}{Definition}{Definitions}

\newaliascnt{assumption}{lemma}

\aliascntresetthe{assumption}
\crefname{assumption}{assumption}{assumptions}
\Crefname{assumption}{Assumption}{Assumptions}

\newaliascnt{question}{lemma}

\aliascntresetthe{question}
\crefname{question}{question}{questions}
\Crefname{question}{Question}{Questions}

\newaliascnt{example}{lemma}

\aliascntresetthe{example}
\crefname{example}{example}{examples}
\Crefname{example}{Example}{Examples}

\theoremstyle{remark}

\newaliascnt{remark}{lemma}

\aliascntresetthe{remark}
\crefname{remark}{remark}{remarks}
\Crefname{remark}{Remark}{Remarks}

\newaliascnt{observation}{lemma}

\aliascntresetthe{observation}
\crefname{observation}{observation}{observations}
\Crefname{observation}{Observation}{Observations}

\newaliascnt{warning}{lemma}

\aliascntresetthe{warning}
\crefname{warning}{warning}{warnings}
\Crefname{warning}{Warning}{Warnings}
\newcommand{\ActiveSet}{\operatorname{ActiveSet}}
\newcommand{\Price}{\operatorname{Price}}
\newcommand{\UpdatePrice}{\operatorname{UpdatePrice}}

\title{A Strongly Polynomial Algorithm for Arctic Auctions} 

\author[1]{Jugal Garg}
\author[2]{Shayan Taherijam}
\author[2]{Vijay V. Vazirani}

\affil[1]{University of Illinois Urbana-Champaign}
\affil[2]{University of California, Irvine}

\date{}

\begin{document}
    \maketitle

\begin{abstract}
Our main contribution is a strongly polynomial algorithm for computing an equilibrium for the Arctic Auction, which is the quasi-linear extension of the linear Fisher market model. We build directly on Orlin's strongly polynomial algorithm for 
the linear Fisher market \cite{Orlin-Fisher}. The first combinatorial polynomial algorithm for the linear Fisher market was based on the primal-dual paradigm \cite{DPSV}. 
This was followed by Orlin's scaling-based algorithms. 

The Arctic Auction \cite{Klemperer2} was developed for the Government of Iceland to allow individuals to exchange blocked offshore assets. It is a variant of the product-mix auction \cite{Klemperer-2008, Klemperer1, Klemperer2} that was designed for, and used by, the Bank of England, to allocate liquidity efficiently across banks pledging heterogeneous collateral of varying quality. Our work was motivated by the fact that banks often need to run Arctic Auctions under many different settings of the parameters in order to home in on the right one,
making it essential to find a time-efficient algorithm for Arctic Auction. 
\end{abstract}

    \input{intro}
    \input{prelim}
    \input{AA_weakly}
    \input{AA_strongly}
 
    \bibliographystyle{alpha}
    \bibliography{refs}
    
    
\end{document}

%% file: intro.tex
\section{Introduction}
\label{sec:intro}

Our main contribution is a strongly polynomial algorithm for computing an equilibrium for the Arctic Auction, a quasi-linear extension of the linear Fisher market. The linear Fisher market is one of the most fundamental models in market equilibrium theory and algorithmic game theory: buyers with budgets compete for divisible goods, and equilibrium prices clear the market, with each buyer spending only on the good with the maximum bang-per-buck. It is well known that equilibrium allocations and prices for this model are captured by the optimal solution to the classical Eisenberg-Gale convex program \cite{eisenberg}. This program, and hence the linear Fisher model, admits rational solutions when all parameters are rational numbers \cite{Va.rational}. 

A central algorithmic milestone for this model was the primal--dual polynomial-time algorithm of Devanur, Papadimitriou, Saberi, and Vazirani~\cite{DPSV}; it extended the primal-dual paradigm from its original setting of LP-duality to convex programs and KKT conditions. A key new idea needed to render the algorithm polynomial time was the notion of balanced flows. This was followed by Orlin's scaling framework and strongly polynomial algorithm~\cite{Orlin-Fisher}. In this paper, we extend Orlin's approach to the Arctic Auction setting.

The Arctic Auction was developed by Klemperer~\cite{Klemperer2} in connection with Iceland's blocked offshore assets, and it is closely related to the product-mix auction~\cite{Klemperer-2008,Klemperer1,Klemperer2}, which was designed for and used by the Bank of England to allocate liquidity across banks pledging heterogeneous collateral. In the Arctic Auction, a bidder may keep part of her budget as money, and each refunded dollar contributes one unit of utility. Thus, unlike in the linear Fisher market, spending the entire budget is not mandatory at equilibrium: if market opportunities are insufficiently attractive, it may be optimal to retain money as a refund. Vazirani~\cite{vazirani2025arctic} recently clarified the structural relation between Arctic Auctions, linear Fisher markets, and rational convex programs. Our work is motivated by the need to compute Arctic Auction equilibria efficiently when the auction must be run repeatedly under different parameter settings.

From an algorithmic viewpoint, allowing refunds creates a significant obstacle. In Orlin's algorithm for the linear Fisher market, any surplus money must eventually be spent. In the Arctic Auction, this is false: once a buyer's maximum bang-per-buck falls to one, the buyer is indifferent between buying equal goods and keeping money. Therefore, the algorithm must distinguish between cash that still needs to be routed through the \emph{equality graph} and cash that should already be treated as a refund. This is exactly where the Fisher-market analysis breaks and where new ideas are needed.

Our first step is to reformulate the weakly polynomial scaling framework so that refunds are part of the state of the algorithm. We represent the state as a triple $(p,x,r)$, where $p$ denotes prices, $x$ the spending variables, and $r$ the refunds. Buyer slack is then measured in terms of \emph{effective cash}, defined by
\[
\bar c_i(x)=e_i-r_i-\sum_{j\in G} x_{ij},
\]
where $e_i$ is the budget of buyer $i$ and $x_{ij}$ is the money spent by buyer $i$ on good $j$. 
This allows us to retain the spending-variable language of Orlin while making the Arctic outside option explicit. With this reformulation, the equality graph, \emph{residual network}, and augmentation machinery are preserved in a natural way, but several definitions and invariants must be modified. In particular, the notion of feasibility and optimality must be stated for $(p,x,r)$ rather than just $(p,x)$, and buyer-terminal augmentation must explicitly transfer one unit of active cash into a refund so that effective cash is preserved correctly.

A second idea is to move the basic-solution viewpoint to the forefront. Under the standard perturbation assumption, the equilibrium support is acyclic. In the Fisher market, each connected support component is pinned down by a budget-equals-price relation. In the Arctic Auction, a component may instead be anchored by a buyer whose outside option is tight. This leads to a slightly different notion of a basic solution, but once the positive support edges of the equilibrium are known, the equilibrium is uniquely recovered. As in Orlin's framework, this makes it possible to identify the exact equilibrium after the scaling parameter becomes sufficiently small and the abundant edges stabilize.

Our main new idea beyond Orlin appears in the strongly polynomial part. We introduce \emph{committed refunds}: when a buyer becomes critical, meaning that her bang-per-buck reaches one, part of her remaining active cash may be committed irrevocably to refund and removed from the effective budget. This yields compressed states in which the algorithm tracks only the money that is still relevant for future price increases and augmentations. Within these compressed states, we define fertile components, modify the price-raising step to include a new stopping event at critical buyers, and adapt Orlin's restart mechanism to the Arctic setting. The restart routine either guarantees the appearance of a new abundant edge after a controlled number of ordinary scaling phases or jumps directly to a much smaller scale while preserving previously discovered abundant edges.

These ideas lead to both weakly and strongly polynomial algorithms. The weakly polynomial algorithm follows the familiar scaling pattern: starting from a coarse scale, it repeatedly restores $\Delta$-optimality, halves the scale, and eventually reconstructs the exact equilibrium from the abundant support. The strongly polynomial algorithm is more subtle. Its analysis shows that progress events are separated by only $O(\log n)$ ordinary phases and that only $O(n)$ such events can occur. This yields a running time of
\[
O\!\left(n^3(m+n\log n)\right),
\]
where $n=|B|+|G|$ and $m=|\{(i,j)\in B\times G : U_{ij}>0\}|$.


The rest of the paper is organized as follows. In Section~\ref{sec:prelim} we define the Arctic Auction with refunds, describe the perturbation assumption, and develop the basic-solution formalism. Section~2 presents the weakly polynomial scaling algorithm and the recovery of the equilibrium from abundant edges. Section~3 develops the strongly polynomial algorithm based on committed refunds and the restart subroutine.

%% file: prelim.tex
\section{Preliminaries}
\label{sec:prelim}

We follow the spending-variable presentation of the linear Fisher market used by Orlin~\cite{Orlin-Fisher}
and the formulation of Arctic Auctions given in~\cite{vazirani2025arctic}. 
Throughout, let $B$ be the set of buyers and $G$ the set of divisible goods.
Each buyer $i\in B$ has budget $e_i>0$, each good $j\in G$ has one unit of supply, and
$U_{ij}\ge 0$ denotes the utility derived by buyer $i$ from one unit of good $j$.
We assume that every buyer values at least one good and every good is valued by at least one buyer, i.e.,
\[
\forall i\in B\ \exists j\in G\text{ such that }U_{ij}>0,
\qquad
\forall j\in G\ \exists i\in B\text{ such that }U_{ij}>0.
\]

For a price vector $p=(p_j)_{j\in G}$ with $p_j>0$, define the \emph{bang-per-buck} of buyer $i$ for good $j$ by
\[
\frac{U_{ij}}{p_j},
\]
and the \emph{maximum bang-per-buck} of buyer $i$ by
\[
\alpha_i(p):=\max_{j\in G}\frac{U_{ij}}{p_j}.
\]
The corresponding \emph{equality graph} is
\[
E(p):=\left\{(i,j)\in B\times G:\frac{U_{ij}}{p_j}=\alpha_i(p)\right\}.
\]
Thus $(i,j)\in E(p)$ if and only if good $j$ is a maximum bang-per-buck good for buyer $i$.

\subsection{Linear Fisher market}

In the linear Fisher market, buyer $i$ spends her entire budget on goods.
Following Orlin~\cite{Orlin-Fisher}, we work with \emph{spending variables}:
for each $(i,j)\in B\times G$, let $x_{ij}\ge 0$ denote the amount of money spent by buyer $i$ on good $j$.
Given a price vector $p$, define the \emph{surplus cash} of buyer $i$ by
\[
c_i(x):=e_i-\sum_{j\in G}x_{ij},
\]
and the \emph{backorder} of good $j$ by
\[
b_j(p,x):=-p_j+\sum_{i\in B}x_{ij}.
\]
Under unit supply, the condition $b_j(p,x)=0$ is equivalent to market clearing for good $j$.

\begin{definition}[Linear Fisher equilibrium]
A pair $(p,x)$ is a \emph{linear Fisher equilibrium} if the following conditions hold:
\begin{enumerate}
    \item \textbf{Budget exhaustion:}
    \[
    c_i(x)=0
    \qquad \forall i\in B.
    \]
    \item \textbf{Market clearing:}
    \[
    b_j(p,x)=0
    \qquad \forall j\in G.
    \]
    \item \textbf{Maximum bang-per-buck support:} if $x_{ij}>0$, then
    \[
    (i,j)\in E(p).
    \]
    \item \textbf{Nonnegativity:}
    \[
    x_{ij}\ge 0 \quad \forall (i,j)\in B\times G,
    \qquad
    p_j\ge 0 \quad \forall j\in G.
    \]
\end{enumerate}
\end{definition}

Equivalently, if $y_{ij}$ denotes the quantity of good $j$ received by buyer $i$, then
$x_{ij}=p_jy_{ij}$ and the condition $b_j(p,x)=0$ is the same as
$\sum_{i\in B}y_{ij}=1$ for each good $j$.

\subsection{Arctic Auction with refunds}

The Arctic Auction is a variant of the linear Fisher market in which a buyer is allowed to keep part of her budget as refunded money~\cite{vazirani2025arctic}. 
If buyer $i$ receives quantity $y_{ij}\ge 0$ of each good $j$ and keeps $r_i\ge 0$ dollars as refund,
then her utility is
\[
\sum_{j\in G}U_{ij}y_{ij}+r_i,
\]
subject to the budget constraint
\[
\sum_{j\in G}p_jy_{ij}+r_i=e_i.
\]
Thus one dollar kept yields one unit of utility.

As above, it is convenient to use spending variables
\[
x_{ij}:=p_jy_{ij}.
\]
For any state $(x,r)$, define the \emph{effective budget} and \emph{effective cash} of buyer $i$ by
\[
\bar e_i:=e_i-r_i,
\qquad
\bar c_i(x):=\bar e_i-\sum_{j\in G}x_{ij}=e_i-r_i-\sum_{j\in G}x_{ij}.
\]
At an Arctic equilibrium we will have $\bar c_i(x)=0$ for every buyer, but later it will be convenient to work with intermediate states in which $\bar c_i(x)$ may be positive.

Using spending variables, buyer $i$'s optimization problem can be written as
\[
\max\left\{\sum_{j\in G}\frac{U_{ij}}{p_j}x_{ij}+r_i:
 x_{ij}\ge 0\ \forall j,\ r_i\ge 0,\ \sum_{j\in G}x_{ij}+r_i=e_i\right\}.
\]
Eliminating $r_i$ gives the equivalent objective
\[
e_i+\sum_{j\in G}\left(\frac{U_{ij}}{p_j}-1\right)x_{ij}.
\]
Therefore the buyer's optimal behavior has the following simple form:
\begin{enumerate}
    \item If $\alpha_i(p)<1$, then buyer $i$ strictly prefers keeping money, hence
    \[
    x_{ij}=0 \quad \forall j\in G,
    \qquad
    r_i=e_i.
    \]
    \item If $\alpha_i(p)>1$, then buyer $i$ spends her full budget and only on equality goods, i.e.,
    \[
    r_i=0,
    \qquad
    x_{ij}>0\implies (i,j)\in E(p).
    \]
    \item If $\alpha_i(p)=1$, then buyer $i$ is indifferent between keeping money and buying equality goods.
    In this case,
    \[
    x_{ij}>0\implies (i,j)\in E(p),
    \]
    and any split between refund and spending on equality goods is optimal.
\end{enumerate}

\begin{definition}[Arctic equilibrium]
A triple $(p,x,r)$ is an \emph{Arctic equilibrium} if the following conditions hold:
\begin{enumerate}
    \item \textbf{Budgets and refunds:}
    \[
    r_i\ge 0,
    \qquad
    \bar c_i(x)=e_i-r_i-\sum_{j\in G}x_{ij}=0
    \qquad \forall i\in B.
    \]
    \item \textbf{Market clearing:}
    \[
    b_j(p,x)=0
    \qquad \forall j\in G.
    \]
    \item \textbf{Maximum bang-per-buck support:} if $x_{ij}>0$, then
    \[
    (i,j)\in E(p).
    \]
    \item \textbf{Refund complementarity:}
    \[
    r_i>0\implies \alpha_i(p)\le 1,
    \qquad
    \alpha_i(p)>1\implies r_i=0.
    \]
\end{enumerate}
\end{definition}

The linear Fisher market is the special case of the Arctic Auction in which every buyer receives zero refund at equilibrium, i.e., $r_i=0$ for all $i\in B$.
\subsection{Generic perturbation}
\label{subsec:generic-perturbation}

Following the standard perturbation framework used in the Fisher-market
literature, we assume that the utilities have been perturbed symbolically so as
to remove degeneracies; see, e.g., \cite{Orlin-Fisher}. We do not need the
explicit form of the perturbation. Throughout the paper we only use the
following two generic properties of the perturbed instance.

For every price vector $p$, let
\[
E(p):=\left\{(i,j)\in B\times G:\frac{U_{ij}}{p_j}=\alpha_i(p)\right\}
\]
denote the equality graph, that is, the subgraph consisting of edges of
maximum bang-per-buck. We assume that the perturbation is chosen so that, for
every price vector $p$:

\begin{enumerate}
    \item the graph $E(p)$ is cycle-free;
    \item each connected component of $E(p)$ contains at most one critical
    buyer, that is, at most one buyer $i$ with $\alpha_i(p) = 1$.
\end{enumerate}

Throughout the paper we work with such a perturbed instance and suppress the
perturbation from the notation.
\subsection{Basic solutions and recovery from the optimal support}

We next define the basic solution associated with a cycle-free support.
Since buyers may receive refunds, a connected component of the support is either budget-balanced, or else it is anchored by a unique buyer whose outside option is tight.

\begin{definition}[Basic solution]
\label{def:basic-solution-arctic}
Let $H\subseteq B\times G$ be cycle-free.
A triple $(p,x,r)$ is called the \emph{basic solution} of $H$ if the following conditions hold:
\begin{enumerate}
    \item If $(i,j)\in H$ and $(i,k)\in H$, then
    \[
    \frac{U_{ij}}{p_j}=\frac{U_{ik}}{p_k}.
    \]
    \item For each connected component $C$ of $H$, exactly one of the following holds:
    \begin{enumerate}
        \item
        \[
        e_{C\cap B}=p_{C\cap G};
        \]
        \item there exists a unique buyer $i_C\in C\cap B$ such that
        \[
        \frac{U_{i_Cj}}{p_j}=1
        \qquad
        \text{for every }(i_C,j)\in H.
        \]
    \end{enumerate}
    \item If $(i,j)\notin H$, then $x_{ij}=0$.
    \item For every buyer $i\in B$,
    \[
    r_i=e_i-\sum_{j\in G}x_{ij}.
    \]
    \item If a buyer $i$ is not one of the anchor buyers $i_C$ from item~2(b), then
    \[
    r_i=0.
    \]
    \item For every good $j\in G$,
    \[
    \sum_{i\in B}x_{ij}=p_j.
    \]
\end{enumerate}
Whenever such a triple exists, we denote it by
\[
(p,x,r)=\operatorname{BasicSolution}(H).
\]
\end{definition}

The next statement shows that once the positive edges of the optimal solution are known, the optimal solution is recovered uniquely.

\begin{theorem}
\label{thm:basic-solution-optimal}
Let $(p^*,x^*,r^*)$ be an Arctic equilibrium, and let
\[
H^*:=\{(i,j)\in B\times G:x^*_{ij}>0\}.
\]
Then
\[
\operatorname{BasicSolution}(H^*)=(p^*,x^*,r^*).
\]
\end{theorem}

\begin{proof}
By the perturbation assumption, the support graph $H^*$ is cycle-free.
We show that $(p^*,x^*,r^*)$ satisfies the defining conditions of \Cref{def:basic-solution-arctic}.

If $(i,j)\in H^*$ and $(i,k)\in H^*$, then both edges carry positive spending, and therefore both are equality edges at $p^*$.
Hence
\[
\frac{U_{ij}}{p_j^*}=\frac{U_{ik}}{p_k^*}.
\]
This gives item~1.

Now consider a connected component $C$ of $H^*$.
If every buyer in $C\cap B$ has zero refund, then
\[
e_{C\cap B}
=
\sum_{i\in C\cap B}\sum_{j\in G}x^*_{ij}
=
\sum_{j\in C\cap G}\sum_{i\in B}x^*_{ij}
=
p^*_{C\cap G},
\]
because there is no spending on edges outside $H^*$ and every good clears.
Thus item~2(a) holds.

Otherwise, some buyer in $C\cap B$ has positive refund.
By the perturbation assumption there is exactly one such buyer; denote it by $i_C$.
By refund complementarity, $\alpha_{i_C}(p^*)\le 1$.
On the other hand, if $(i_C,j)\in H^*$ then $x^*_{i_Cj}>0$, so $(i_C,j)\in E(p^*)$ and therefore
\[
\frac{U_{i_Cj}}{p_j^*}=\alpha_{i_C}(p^*).
\]
Since $x^*_{i_Cj}>0$ and $r_{i_C}^*>0$ occur simultaneously, buyer $i_C$ must be indifferent between spending on equality goods and keeping money, which forces $\alpha_{i_C}(p^*)=1$.
Hence
\[
\frac{U_{i_Cj}}{p_j^*}=1
\qquad
\text{for every }(i_C,j)\in H^*.
\]
Thus item~2(b) holds.

Item~3 is immediate from the definition of $H^*$.
Item~4 is exactly the buyer budget identity in the Arctic equilibrium.
For item~5, if $i$ is not an anchor buyer then $r_i^*=0$ by construction of the anchor buyer in its component.
Item~6 is exactly market clearing.

Therefore $(p^*,x^*,r^*)$ satisfies the equations defining the basic solution of $H^*$.
Since each connected component of $H^*$ is a tree and each component has exactly one scale-fixing condition from item~2, these equations determine $(p,x,r)$ uniquely.
Hence
\[
\operatorname{BasicSolution}(H^*)=(p^*,x^*,r^*).
\]
\end{proof}

We next record the lower bound on positive spending that will later allow us to recognize the optimal support once $\Delta$ becomes sufficiently small.

\begin{lemma}
\label{lem:min-positive-flow}
If $x^*_{ij}>0$, then
\[
x^*_{ij}>\frac{1}{D},
\]
where $D:=n(U_{\max})^n$ and $n:=|B|+|G|$.
\end{lemma}

\begin{proof}
By \Cref{thm:basic-solution-optimal}, the triple $(p^*,x^*,r^*)$ is obtained from the linear system defining $\operatorname{BasicSolution}(H^*)$.
After clearing denominators, this system has integral coefficients, and every coefficient has absolute value at most $U_{\max}$.
By Cramer's rule, every coordinate of $(p^*,x^*,r^*)$ is a rational number whose denominator is less than $D=n(U_{\max})^n$.
Therefore every positive coordinate is at least $1/D$.
In particular, if $x^*_{ij}>0$, then $x^*_{ij}>1/D$.
\end{proof}

%% file: AA_weakly.tex
\section{A \texorpdfstring{$\Delta$}{Delta}-scaling algorithm for the Arctic Auction}
\label{sec:modified-price-augment}

We adapt the $\Delta$-scaling framework of Orlin~\cite{Orlin-Fisher} to the Arctic Auction.
Throughout this section, we assume that the perturbation from \Cref{sec:prelim} is in force.
We work with states $(p,x,r)$, where $p$ is a price vector, $x$ is a spending vector, and $r$ is a refund vector.
We also assume that the budgets and utilities are integral.
Let
\[
n:=|B|+|G|,
\qquad
U_{\max}:=\max\{U_{ij}:U_{ij}>0\},
\qquad
 e_{\max}:=\max_{i\in B}e_i,
\qquad
D:=n(U_{\max})^n.
\]
For sets $S\subseteq B$ and $T\subseteq G$, we write
\[
e_S:=\sum_{i\in S}e_i,
\qquad
p_T:=\sum_{j\in T}p_j.
\]
Recall that in a state $(p,x,r)$,
\[
\bar e_i:=e_i-r_i,
\qquad
\bar c_i(x):=\bar e_i-\sum_{j\in G}x_{ij}=e_i-r_i-\sum_{j\in G}x_{ij},
\]
\[
b_j(p,x):=-p_j+\sum_{i\in B}x_{ij},
\qquad
\alpha_i(p):=\max_{j\in G}\frac{U_{ij}}{p_j},
\]
and
\[
E(p):=\left\{(i,j)\in B\times G:\frac{U_{ij}}{p_j}=\alpha_i(p)\right\}.
\]
Since a buyer with $\alpha_i(p)\le 1$ weakly prefers money to any non-equality purchase, any remaining effective cash of such a buyer may be safely committed to refund.
For this reason we use the stronger stopping rule below and require small effective cash for \emph{all} buyers at the end of a scaling phase.

\begin{definition}[$\Delta$-feasible solution]
\label{def:delta-feasible-arctic}
Let $\Delta>0$.
A state $(p,x,r)$ is called \emph{$\Delta$-feasible} if the following conditions hold:
\begin{enumerate}
    \item $r_i\ge 0$ and $\bar c_i(x)\ge 0$ for every buyer $i\in B$;
    \item if $p_j>p_j^0$, then $0\le b_j(p,x)\le \Delta$ for every good $j\in G$;
    \item if $x_{ij}>0$, then $(i,j)\in E(p)$ and $x_{ij}$ is a multiple of $\Delta$;
    \item $x_{ij}\ge 0$ for all $(i,j)\in B\times G$ and $p_j\ge 0$ for all $j\in G$.
\end{enumerate}
\end{definition}

\begin{definition}[$\Delta$-optimal solution]
\label{def:delta-optimal-arctic}
A $\Delta$-feasible state $(p,x,r)$ is called \emph{$\Delta$-optimal} if
\[
\bar c_i(x)<\Delta
\qquad
\text{for every buyer }i\in B.
\]
\end{definition}

\subsection{Initialization and scaling invariants}

We now define the initial parameters of the scaling algorithm.
For each buyer $i\in B$, let
\[
U_{iG}:=\sum_{j\in G}U_{ij}.
\]
Set
\[
\Delta_0:=e_{\max},
\qquad
\rho_{ij}:=\frac{U_{ij}e_i}{nU_{iG}},
\qquad
p_j^0:=\max_{i\in B}\rho_{ij},
\qquad
x_{ij}^0:=0,
\qquad
r_i^0:=0.
\]

The role of the factor $1/n$ in the definition of $\rho_{ij}$ is to guarantee that the initial prices are sufficiently small, while the choice $\Delta_0=e_{\max}$ ensures that the initial effective cash is already controlled at the coarsest scale.

\begin{lemma}
\label{lem:initial-feasible}
The initial state $(p^0,x^0,r^0)$ is $\Delta_0$-feasible.
\end{lemma}

\begin{proof}
Since $x^0=0$ and $r^0=0$, we have
\[
\bar c_i(x^0)=e_i\ge 0
\qquad
\text{for all }i\in B.
\]
Because $p=p^0$, the condition ``if $p_j>p_j^0$'' in the definition of $\Delta$-feasibility is vacuous.
Moreover, there is no positive spending, so the support and integrality condition is also vacuous.
Finally, $x_{ij}^0\ge 0$, $r_i^0=0$, and $p_j^0\ge 0$.
Thus $(p^0,x^0,r^0)$ is $\Delta_0$-feasible.
\end{proof}

\begin{definition}[Residual network and active set]
\label{def:residual-network-arctic}
Let $(p,x,r)$ be a $\Delta$-feasible state.
The residual network $N(p,x,r)$ is the directed graph with node set $B\cup G$ defined as follows:
\begin{enumerate}
    \item for every equality edge $(i,j)\in E(p)$, the network contains the forward arc $(i,j)$;
    \item for every pair $(i,j)$ with $x_{ij}>0$, the network contains the backward arc $(j,i)$.
\end{enumerate}
If $\rho\in B$, then $\ActiveSet(p,x,r,\rho)$ is the set of nodes reachable from $\rho$ in $N(p,x,r)$.
\end{definition}

\begin{definition}[$\Delta$-augmentation]
\label{def:delta-augmentation-arctic}
Let $(p,x,r)$ be a $\Delta$-feasible state, and let
\[
P=(v_0,v_1,\dots,v_t)
\]
be a directed path in $N(p,x,r)$ whose first node $v_0$ is a buyer.
The $\Delta$-augmentation of $x$ along $P$ is the allocation $x'$ defined by
\[
x'_{ij}=
\begin{cases}
x_{ij}+\Delta, & \text{if }(i,j)\text{ appears on }P\text{ as a forward arc},\\[1mm]
x_{ij}-\Delta, & \text{if }(j,i)\text{ appears on }P\text{ as a backward arc},\\[1mm]
x_{ij}, & \text{otherwise}.
\end{cases}
\]
\end{definition}

\subsection{The modified price-and-augment step}

Fix a $\Delta$-feasible state $(p,x,r)$ and a root buyer $\rho\in B$.
Let
\[
A=\ActiveSet(p,x,r,\rho).
\]
For $q\ge 1$, define $\Price(p,x,r,\rho,q)$ to be the price vector $\widehat p$ given by
\[
\widehat p_j=
\begin{cases}
q\,p_j, & \text{if }j\in A\cap G,\\[1mm]
p_j, & \text{if }j\notin A\cap G.
\end{cases}
\]
We then define $\UpdatePrice^{\star}(p,x,r,\rho)$ to be the price vector obtained by choosing the maximum value of $q\ge 1$ such that $(\widehat p,x,r)$ remains $\Delta$-feasible and at least one of the following events occurs:
\begin{enumerate}
    \item a new equality edge appears;
    \item some active good $j\in A\cap G$ satisfies $b_j(\widehat p,x)\le 0$;
    \item some active buyer $i\in A\cap B$ satisfies $\alpha_i(\widehat p)=1$.
\end{enumerate}
The third event is the new feature in the Arctic setting.

\begin{lemma}
\label{lem:update-price-star-feasible}
Let $(p,x,r)$ be $\Delta$-feasible and let $\rho\in B$.
Then
\[
\bigl(\UpdatePrice^{\star}(p,x,r,\rho),x,r\bigr)
\]
is again $\Delta$-feasible.
\end{lemma}

\begin{proof}
This follows directly from the definition of $\UpdatePrice^{\star}$.
\end{proof}

\begin{algorithm}[H]
\caption{\textsc{PriceAndAugmentWithCriticalBuyer}$(p,x,r,\Delta)$}
\label{alg:arctic-price-augment}
\begin{algorithmic}[1]
\Require A $\Delta$-feasible state $(p,x,r)$
\State Choose a buyer $\rho\in B$ such that $\bar c_\rho(x)\ge \Delta$ and $\alpha_\rho(p)>1$
\State $A\leftarrow \ActiveSet(p,x,r,\rho)$
\While{there is no active good $j\in A\cap G$ with $b_j(p,x)\le 0$ and no active buyer $i\in A\cap B$ with $\alpha_i(p)=1$}
    \State $p\leftarrow \UpdatePrice^{\star}(p,x,r,\rho)$
    \State Recompute $E(p)$, $N(p,x,r)$, $\alpha(p)$, $b(p,x)$, and $A\leftarrow \ActiveSet(p,x,r,\rho)$
\EndWhile
\If{there exists an active buyer $i\in A\cap B$ with $\alpha_i(p)=1$}
    \State Choose a directed path $P$ in $N(p,x,r)$ from $\rho$ to $i$
    \State Replace $x$ by the $\Delta$-augmentation of $x$ along $P$
    \State $r_i\leftarrow r_i+\Delta$
\Else
    \State Choose an active good $j\in A\cap G$ with $b_j(p,x)\le 0$
    \State Choose a directed path $P$ in $N(p,x,r)$ from $\rho$ to $j$
    \State Replace $x$ by the $\Delta$-augmentation of $x$ along $P$
\EndIf
\State Recompute $\bar c(x)$ and $b(p,x)$
\State \Return $(p,x,r)$
\end{algorithmic}
\end{algorithm}

Given a $\Delta$-optimal state $(p,x,r)$, we define its \emph{halving-and-repair step} as follows.
Set $\Delta':=\Delta/2$.
For each good $j\in G$ with $b_j(p,x)>\Delta'$, choose a buyer $i(j)$ with $x_{i(j)j}>0$ and replace
\[
x_{i(j)j}\leftarrow x_{i(j)j}-\Delta'.
\]
Let $\widetilde x$ denote the resulting allocation, and keep $p$ and $r$ unchanged.

\begin{algorithm}[H]
\caption{\textsc{ScalingAlgorithm}$(e,U)$}
\label{alg:scaling-arctic}
\begin{algorithmic}[1]
\State Compute $\Delta_0$, $p^0$, $x^0$, and $r^0$ as above
\State $\Delta\leftarrow \Delta_0$, $p\leftarrow p^0$, $x\leftarrow x^0$, $r\leftarrow r^0$
\While{true}
    \While{$(p,x,r)$ is not $\Delta$-optimal}
        \If{there exists a buyer $i\in B$ with $\alpha_i(p)\le 1$ and $\bar c_i(x)\ge \Delta$}
            \State $r_i\leftarrow r_i+\Delta$
        \Else
            \State $(p,x,r)\leftarrow \textsc{PriceAndAugmentWithCriticalBuyer}(p,x,r,\Delta)$
        \EndIf
    \EndWhile
    \If{$\Delta<1/(8nD)$}
        \State $E_\Delta\leftarrow \{(i,j)\in B\times G:x_{ij}>4n\Delta\}$
        \State \Return $\operatorname{BasicSolution}(E_\Delta)$
    \EndIf
    \State Replace $(p,x,r)$ by $(p,\widetilde x,r)$ obtained from the halving-and-repair step
    \State $\Delta\leftarrow \Delta/2$
\EndWhile
\end{algorithmic}
\end{algorithm}

The rest of this section proves the correctness of \Cref{alg:arctic-price-augment,alg:scaling-arctic}.

\begin{lemma}[Refund step]
\label{lem:refund-step}
Let $(p,x,r)$ be a $\Delta$-feasible state, and suppose that buyer $i\in B$ satisfies $\alpha_i(p)\le 1$ and $\bar c_i(x)\ge \Delta$.
Let $r'$ be obtained from $r$ by replacing $r_i$ with $r_i+\Delta$.
Then $(p,x,r')$ is again $\Delta$-feasible and
\[
\bar c_i'(x)=\bar c_i(x)-\Delta.
\]
All other buyers keep the same effective cash.
\end{lemma}

\begin{proof}
Increasing $r_i$ by $\Delta$ decreases $\bar e_i$ and hence $\bar c_i(x)$ by exactly $\Delta$, while leaving prices, allocations, and therefore all backorders unchanged.
Since $\bar c_i(x)\ge \Delta$, the new effective cash remains nonnegative.
All other conditions in \Cref{def:delta-feasible-arctic} are unchanged.
\end{proof}

\begin{lemma}[Buyer-terminal augmentation]
\label{lem:buyer-terminal-augmentation}
Let $(p,x,r)$ be a $\Delta$-feasible state, let $\rho\in B$ satisfy $\bar c_\rho(x)\ge \Delta$, and let $i\in B$ be reachable from $\rho$ in the residual network $N(p,x,r)$.
Let $P$ be a directed path from $\rho$ to $i$, let $x'$ be obtained from $x$ by a $\Delta$-augmentation along $P$, and let $r'$ be obtained from $r$ by replacing $r_i$ with $r_i+\Delta$.
Then:
\begin{enumerate}
    \item $b_j(p,x')=b_j(p,x)$ for every good $j\in G$;
    \item $\bar c_\ell'(x')=\bar c_\ell(x)$ for every buyer $\ell\in B\setminus\{\rho\}$;
    \item
    \[
    \bar c_\rho'(x')=\bar c_\rho(x)-\Delta.
    \]
\end{enumerate}
Consequently, $(p,x',r')$ is again $\Delta$-feasible.
\end{lemma}

\begin{proof}
Along a directed buyer-to-buyer path in $N(p,x,r)$, every internal good has one incoming forward arc and one outgoing backward arc, so the total spending into each internal good is unchanged.
Similarly, every internal buyer has one incident variable increased by $\Delta$ and one incident variable decreased by $\Delta$, so its total spending is unchanged.
The root buyer $\rho$ has one additional outgoing increment of $\Delta$.
The terminal buyer $i$ has one incoming decrement of $\Delta$, but we also increase her refund by $\Delta$.
Therefore the effective cash of every buyer except $\rho$ is unchanged, while the effective cash of $\rho$ decreases by exactly $\Delta$.
This proves the three identities.

Since every backward arc corresponds to a positive allocation and every positive allocation is a multiple of $\Delta$, no variable becomes negative after the augmentation.
Moreover, every positive allocation remains supported on an equality edge and remains a multiple of $\Delta$.
Hence all conditions of \Cref{def:delta-feasible-arctic} are preserved.
\end{proof}

\begin{lemma}[Good-terminal augmentation]
\label{lem:good-terminal-augmentation}
Let $(p,x,r)$ be a $\Delta$-feasible state, let $\rho\in B$ satisfy $\bar c_\rho(x)\ge \Delta$, and let $j\in G$ be reachable from $\rho$ in $N(p,x,r)$ with $b_j(p,x)\le 0$.
Let $P$ be a directed path from $\rho$ to $j$, and let $x'$ be obtained from $x$ by a $\Delta$-augmentation along $P$.
Then:
\begin{enumerate}
    \item $b_h(p,x')=b_h(p,x)$ for every good $h\in G\setminus\{j\}$;
    \item $b_j(p,x')=b_j(p,x)+\Delta$;
    \item $\bar c_\rho(x')=\bar c_\rho(x)-\Delta$ and $\bar c_\ell(x')=\bar c_\ell(x)$ for every buyer $\ell\in B\setminus\{\rho\}$.
\end{enumerate}
Consequently, $(p,x',r)$ is again $\Delta$-feasible.
\end{lemma}

\begin{proof}
Every internal good on the path has one incoming forward arc and one outgoing backward arc, so its total incoming spending does not change.
The terminal good $j$ has one incoming forward arc and no outgoing backward arc, and therefore its incoming spending increases by exactly $\Delta$.
This proves the first two claims.

Every internal buyer has one incident variable increased by $\Delta$ and one incident variable decreased by $\Delta$, so its total spending is unchanged.
The root buyer $\rho$ has one additional outgoing increment of $\Delta$, hence $\bar c_\rho(x')=\bar c_\rho(x)-\Delta$, and all other buyers keep the same effective cash.

For goods $h\neq j$, the backorder is unchanged.
For the terminal good $j$, we have
\[
b_j(p,x')=b_j(p,x)+\Delta\le \Delta,
\]
because $b_j(p,x)\le 0$.
All other conditions in \Cref{def:delta-feasible-arctic} are preserved exactly as in the previous lemma.
\end{proof}

To measure progress within a scaling phase, we use the potential
\[
\Phi_\Delta(p,x,r):=
\sum_{i\in B}\left\lfloor\frac{\bar c_i(x)}{\Delta}\right\rfloor.
\]

\begin{lemma}[Potential drop in the refund step]
\label{lem:potential-drop-refund}
Let $(p,x,r)$ be a $\Delta$-feasible state.
Suppose that buyer $i\in B$ satisfies $\alpha_i(p)\le 1$ and $\bar c_i(x)\ge \Delta$.
If $(p,x,r')$ is obtained by increasing $r_i$ by $\Delta$, then
\[
\Phi_\Delta(p,x,r')=\Phi_\Delta(p,x,r)-1.
\]
\end{lemma}

\begin{proof}
By \Cref{lem:refund-step}, only buyer $i$ changes effective cash, and her effective cash decreases by exactly $\Delta$.
Hence
\[
\left\lfloor\frac{\bar c_i'(x)}{\Delta}\right\rfloor
=
\left\lfloor\frac{\bar c_i(x)-\Delta}{\Delta}\right\rfloor
=
\left\lfloor\frac{\bar c_i(x)}{\Delta}\right\rfloor-1,
\]
while all other terms remain unchanged.
\end{proof}

\begin{lemma}[Potential drop in the buyer-terminal case]
\label{lem:potential-drop-buyer}
Let $(p,x,r)$ be $\Delta$-feasible.
Suppose that $\rho\in B$ satisfies $\bar c_\rho(x)\ge \Delta$, and that
$i\in \ActiveSet(p,x,r,\rho)\cap B$ satisfies $\alpha_i(p)=1$.
If $(p,x',r')$ is obtained by augmenting along a directed path from $\rho$ to $i$ and then increasing $r_i$ by $\Delta$, then
\[
\Phi_\Delta(p,x',r')=\Phi_\Delta(p,x,r)-1.
\]
\end{lemma}

\begin{proof}
By \Cref{lem:buyer-terminal-augmentation}, only the effective cash of $\rho$ changes, and it decreases by exactly $\Delta$.
Thus
\[
\left\lfloor\frac{\bar c_\rho'(x')}{\Delta}\right\rfloor
=
\left\lfloor\frac{\bar c_\rho(x)-\Delta}{\Delta}\right\rfloor
=
\left\lfloor\frac{\bar c_\rho(x)}{\Delta}\right\rfloor-1,
\]
while all other terms remain unchanged.
\end{proof}

\begin{lemma}[Potential drop in the good-terminal case]
\label{lem:potential-drop-good}
Let $(p,x,r)$ be $\Delta$-feasible.
Suppose that $\rho\in B$ satisfies $\bar c_\rho(x)\ge \Delta$, and that there is an active good
$j\in \ActiveSet(p,x,r,\rho)\cap G$ with $b_j(p,x)\le 0$.
If $(p,x',r)$ is obtained by augmenting along a directed path from $\rho$ to $j$, then
\[
\Phi_\Delta(p,x',r)=\Phi_\Delta(p,x,r)-1.
\]
\end{lemma}

\begin{proof}
By \Cref{lem:good-terminal-augmentation}, the only buyer whose effective cash changes is the root buyer $\rho$, and
\[
\bar c_\rho(x')=\bar c_\rho(x)-\Delta.
\]
Hence the contribution of $\rho$ to the potential decreases by exactly one, and all other terms remain unchanged.
\end{proof}

\begin{proposition}
\label{prop:modified-phase}
Every iteration of the inner loop of \Cref{alg:scaling-arctic} transforms a $\Delta$-feasible state into another $\Delta$-feasible state and decreases the potential $\Phi_\Delta$ by exactly one.
Consequently, the inner loop terminates after finitely many iterations and returns a $\Delta$-optimal state.
\end{proposition}

\begin{proof}
If the iteration is a refund step, feasibility preservation is given by \Cref{lem:refund-step} and the potential drop is given by \Cref{lem:potential-drop-refund}.
If the iteration calls \Cref{alg:arctic-price-augment}, feasibility preservation under price updates is given by \Cref{lem:update-price-star-feasible}, feasibility preservation under augmentations is given by \Cref{lem:buyer-terminal-augmentation,lem:good-terminal-augmentation}, and the potential drop is given by \Cref{lem:potential-drop-buyer,lem:potential-drop-good}.
Since $\Phi_\Delta$ is a nonnegative integer, the inner loop must terminate.
At termination there is no buyer $i$ with $\bar c_i(x)\ge \Delta$, which is exactly the condition in \Cref{def:delta-optimal-arctic}.
\end{proof}

\subsection{Passing from one scale to the next}

Once a $\Delta$-optimal state has been reached, the algorithm either stops if the support has already been identified, or else passes to the next scale.

\begin{proposition}[Halving the scaling parameter]
\label{prop:halving}
Let $(p,x,r)$ be a $\Delta$-optimal state, and set $\Delta':=\Delta/2$.
For each good $j\in G$ with $b_j(p,x)>\Delta'$, choose a buyer $i(j)$ with $x_{i(j)j}>0$ and replace
\[
x_{i(j)j}\leftarrow x_{i(j)j}-\Delta'.
\]
Let $\widetilde x$ denote the resulting allocation.
Then $(p,\widetilde x,r)$ is $\Delta'$-feasible.
\end{proposition}

\begin{proof}
We verify the four conditions of \Cref{def:delta-feasible-arctic}.

First, every refund remains nonnegative and every buyer's effective cash remains nonnegative, because the repair step only decreases some spending variables.
Second, if $p_j>p_j^0$, then initially $0\le b_j(p,x)\le \Delta$.
If $b_j(p,x)\le \Delta'$, the backorder is unchanged.
If $b_j(p,x)>\Delta'$, then after the repair step the new backorder is
\[
b_j(p,\widetilde x)=b_j(p,x)-\Delta'\in(0,\Delta'].
\]
Hence the backorder condition holds for every good.

Third, every positive allocation in $\widetilde x$ remains supported on an equality edge, since the repair step does not introduce any new positive variable.
Moreover, each allocation in $x$ is a multiple of $\Delta=2\Delta'$, and therefore also a multiple of $\Delta'$.
Subtracting $\Delta'$ preserves the property of being a multiple of $\Delta'$.

Finally, nonnegativity holds because every repaired variable satisfies
\[
x_{i(j)j}\ge \Delta=2\Delta'>\Delta',
\]
as $x_{i(j)j}$ is a positive multiple of $\Delta$.
Thus $(p,\widetilde x,r)$ is $\Delta'$-feasible.
\end{proof}

\begin{proposition}[Potential bound at the beginning of a phase]
\label{prop:phase-potential-bound}
At the beginning of every $\Delta$-scaling phase,
\[
\Phi_\Delta(p,x,r)\le n.
\]
Consequently, every $\Delta$-scaling phase performs at most $n$ iterations of the inner loop of \Cref{alg:scaling-arctic}.
\end{proposition}

\begin{proof}
For the initial phase, $\Delta=\Delta_0=e_{\max}$ and $(x,r)=(x^0,r^0)$, so
\[
\bar c_i(x^0)=e_i\le e_{\max}=\Delta_0
\qquad
\text{for every }i\in B.
\]
Hence
\[
\Phi_{\Delta_0}(p^0,x^0,r^0)\le |B|<n.
\]

Now consider any later $\Delta$-phase.
It starts from a $2\Delta$-optimal state, followed by the repair step of \Cref{prop:halving}.
Before the repair step, every buyer has effective cash strictly less than $2\Delta$, so the contribution of each buyer to $\Phi_\Delta$ is at most $1$.
Thus the potential is at most $|B|$ at that point.
Each repaired good increases the effective cash of at most one buyer by exactly $\Delta$, and therefore increases the potential by at most $1$.
Since there are at most $|G|$ repaired goods, the total increase is at most $|G|$.
Therefore
\[
\Phi_\Delta(p,x,r)\le |B|+|G|=n
\]
at the beginning of the new phase.
The second statement follows from \Cref{prop:modified-phase}.
\end{proof}

\subsection{Abundant edges and recovery of the equilibrium}

\begin{definition}[$\Delta$-abundant edge]
\label{def:abundant-edge}
Let $(p,x,r)$ be the state at the beginning of a $\Delta$-scaling phase.
An edge $(i,j)$ is called $\Delta$-abundant if
\[
x_{ij}\ge 3n\Delta.
\]
\end{definition}

The previous subsection implies that within a single phase, each spending variable can change only by a controlled amount.

\begin{lemma}
\label{lem:per-phase-change}
Let $(p,x,r)$ be the state at the beginning of a $\Delta$-scaling phase, and let $(p',x',r')$ be the state at the end of that phase, before the repair step of \Cref{prop:halving}.
Then for every edge $(i,j)$,
\[
|x'_{ij}-x_{ij}|\le n\Delta.
\]
\end{lemma}

\begin{proof}
By \Cref{prop:phase-potential-bound}, the phase performs at most $n$ iterations.
Each iteration changes any fixed coordinate by at most $\Delta$ in absolute value.
Hence the total change of any coordinate during the phase is at most $n\Delta$.
\end{proof}

\begin{lemma}
\label{lem:abundant-persists}
If an edge is $\Delta$-abundant at the beginning of a $\Delta$-phase, then it is $(\Delta/2)$-abundant at the beginning of the next phase.
In particular, once an edge becomes abundant, it remains abundant in all later phases.
\end{lemma}

\begin{proof}
Let $(i,j)$ be $\Delta$-abundant at the beginning of the phase, so $x_{ij}\ge 3n\Delta$.
By \Cref{lem:per-phase-change}, at the end of the phase, before repair, its value is at least
\[
3n\Delta-n\Delta=2n\Delta.
\]
The repair step from \Cref{prop:halving} decreases any fixed edge by at most $\Delta/2$.
Therefore at the beginning of the next phase the same edge carries at least
\[
2n\Delta-\frac{\Delta}{2}.
\]
Since
\[
2n\Delta-\frac{\Delta}{2}>\frac{3n\Delta}{2}=3n\left(\frac{\Delta}{2}\right),
\]
the edge is $(\Delta/2)$-abundant.
Iterating this argument proves the second statement.
\end{proof}

To identify the optimal support, we compare the current solution to the limiting equilibrium.
The next lemma gives a uniform bound on the total future change of a coordinate.

\begin{lemma}
\label{lem:future-change-bound}
Let $(p,x,r)$ be the state at the beginning of a $\Delta$-scaling phase, and let $(p^*,x^*,r^*)$ be the limiting equilibrium reached by the scaling algorithm.
Then for every edge $(i,j)$,
\[
|x^*_{ij}-x_{ij}|<4n\Delta.
\]
\end{lemma}

\begin{proof}
During the current phase, the change in $x_{ij}$ is at most $n\Delta$ by \Cref{lem:per-phase-change}.
The subsequent repair step changes the same coordinate by at most $\Delta/2$.
At the next phase, whose parameter is $\Delta/2$, the same argument gives an additional change of at most
\[
n\frac{\Delta}{2}+\frac{\Delta}{4},
\]
and so on.
Summing the resulting geometric series yields
\[
|x^*_{ij}-x_{ij}|
\le
\sum_{t=0}^{\infty}\left(n\frac{\Delta}{2^t}+\frac{\Delta}{2^{t+1}}\right)
<
2n\Delta+\Delta
<
4n\Delta.
\]
\end{proof}

We can now state the termination criterion.

\begin{theorem}
\label{thm:recover-support}
Let $(p,x,r)$ be the state at the beginning of a $\Delta$-scaling phase.
If
\[
\Delta<\frac{1}{8nD},
\]
and
\[
E_\Delta:=\{(i,j)\in B\times G:x_{ij}>4n\Delta\},
\]
then
\[
E_\Delta=H^*,
\]
where $H^*=\{(i,j):x^*_{ij}>0\}$ is the support of the equilibrium.
Consequently,
\[
\operatorname{BasicSolution}(E_\Delta)=(p^*,x^*,r^*).
\]
\end{theorem}

\begin{proof}
First suppose that $(i,j)\in E_\Delta$.
Then $x_{ij}>4n\Delta$, and by \Cref{lem:future-change-bound},
\[
x^*_{ij}>x_{ij}-4n\Delta>0.
\]
Thus $(i,j)\in H^*$.

Conversely, suppose that $(i,j)\notin E_\Delta$.
Then $x_{ij}\le 4n\Delta$.
If $(i,j)\in H^*$, then \Cref{lem:min-positive-flow} gives
\[
x^*_{ij}>\frac{1}{D}.
\]
On the other hand, \Cref{lem:future-change-bound} gives
\[
x^*_{ij}<x_{ij}+4n\Delta\le 8n\Delta<\frac{1}{D},
\]
which is a contradiction.
Therefore $(i,j)\notin H^*$.

We have shown that $E_\Delta=H^*$.
The final claim now follows from \Cref{thm:basic-solution-optimal}.
\end{proof}

Therefore the algorithm finds the equilibrium by running scaling phases until the parameter becomes smaller than $1/(8nD)$.
At that point, the set of edges carrying more than $4n\Delta$ units of spending is exactly the support of the equilibrium, and the equilibrium is recovered by computing the corresponding basic solution.

%% file: AA_strongly.tex
\section{A strongly polynomial algorithm}
\label{sec:strongly-poly}

In this section we develop a strongly polynomial algorithm for the Arctic Auction with linear utilities.
The overall architecture follows the scaling-and-restart framework introduced in~\cite{Orlin-Fisher} for the linear Fisher market: ordinary scaling phases alternate with a restart subroutine that, whenever the current scale fails to force a new support edge, jumps directly to a much smaller scaling parameter.
The Arctic setting, however, contains buyers whose bang-per-buck is exactly one; such buyers may keep part of their budget as refund and therefore break the Fisher identity ``all remaining money must eventually be spent.''
The modification introduced in this section handles this obstruction by \emph{committing part of a critical buyer's remaining money to refund} as soon as she becomes critical; the committed amount is removed from the active budget and restored only at the end of the algorithm.
Combined with a suitably modified price-raising subroutine, this yields a strongly polynomial algorithm that runs in
\[
O\!\left((n^2\log n)(m+n\log n)\right)
\]
time, where $n=|B|+|G|$ and $m=|\{(i,j)\in B\times G:U_{ij}>0\}|$.

Throughout this section we retain the perturbation from \Cref{sec:prelim}, so that the equality graph of the equilibrium is a forest and the equilibrium is unique.
We also retain the $\Delta$-scaling machinery of \Cref{sec:modified-price-augment}: the procedures \textsc{PriceAndAugmentWithCriticalBuyer} and the halving step of \Cref{prop:halving}, together with the bound of \Cref{prop:phase-potential-bound} that every $\Delta$-phase performs at most $n$ augmentations.

\subsection{Committed refunds and effective budgets}
\label{subsec:committed-refunds}

The algorithm operates on a compressed representation of the instance, which keeps track of money that has already been irrevocably declared to be refunded to its owner.

\begin{definition}[Compressed Arctic state]
\label{def:committed-refunds}
A \emph{compressed Arctic state} is a triple $(p,x,r)$, where $p$ is a price vector, $x$ is a spending vector, and $r=(r_i)_{i\in B}\in\mathbb R_{\ge 0}^B$ is a vector of \emph{committed refunds}.
The corresponding \emph{effective budget} and \emph{effective cash} of buyer $i$ are
\[
\bar e_i:=e_i-r_i,
\qquad
\bar c_i(x):=\bar e_i-\sum_{j\in G}x_{ij}=c_i(x)-r_i.
\]
For a set of nodes $H\subseteq B\cup G$, define the \emph{effective component surplus}
\[
s(p,r,H):=\bar e_{H\cap B}-p_{H\cap G}.
\]
\end{definition}

A compressed Arctic state $(p,x,r)$ is called \emph{$\Delta$-feasible}
(respectively \emph{$\Delta$-optimal}) if the pair $(p,x)$ satisfies the
conditions of \Cref{def:delta-feasible-arctic}
(respectively \Cref{def:delta-optimal-arctic}) with $e_i$ replaced
everywhere by $\bar e_i$.
When $r\equiv 0$, all notions reduce to those of
\Cref{sec:modified-price-augment}.

For a cycle-free set $H\subseteq B\times G$ and an effective-budget vector
$\widehat e\in \mathbb R_{>0}^B$, let
\[
\operatorname{BasicSolution}_{\widehat e}(H)
\]
denote the triple obtained from \Cref{def:basic-solution-arctic} after replacing
every occurrence of $e_i$ by $\widehat e_i$.
In a compressed state $(p,x,r)$ we will abbreviate this as
\[
\operatorname{BasicSolution}_{e-r}(H).
\]

\begin{lemma}[Monotonicity of bang-per-buck]
\label{lem:alpha-monotone-strong}
If $p'\ge p$ coordinatewise, then $\alpha_i(p')\le \alpha_i(p)$ for every buyer $i$.
In particular, once a buyer satisfies $\alpha_i(p)=1$, she satisfies $\alpha_i(p')\le 1$ at every later state of the algorithm.
\end{lemma}

\begin{proof}
For every good $j$ we have $U_{ij}/p'_j\le U_{ij}/p_j$, and taking the maximum over $j$ yields the claim.
\end{proof}

\begin{lemma}[Safe refund commitment]
\label{lem:safe-refund-commitment}
Let $(p,x,r)$ be a compressed Arctic state, and suppose that buyer $i\in B$
satisfies
\[
\alpha_i(p)=1,
\qquad
0\le \delta\le \bar c_i(x).
\]
Let $r'$ be obtained from $r$ by replacing $r_i$ with $r_i+\delta$.
Then the following hold.
\begin{enumerate}
    \item $(p,x,r')$ is again a compressed Arctic state, and
    \[
    \bar c_i'(x)=\bar c_i(x)-\delta\ge 0.
    \]
    \item The equality graph $E(p)$, the set of positive-spending edges, and the set
    of abundant edges are all unchanged.
    \item Let $q\ge p$ coordinatewise, and let $(q,z,\rho)$ be an Arctic equilibrium
    of the compressed instance with budget vector $e-r'$.
    Define
    \[
    \widehat \rho:=\rho+r'.
    \]
    Then $(q,z,\widehat \rho)$ is an Arctic equilibrium of the original instance with
    budget vector $e$.
\end{enumerate}
\end{lemma}

\begin{proof}
Part~(1) is immediate: only buyer $i$'s effective budget changes, and it changes by
exactly $\delta$, so
\[
\bar c_i'(x)=e_i-(r_i+\delta)-\sum_{j\in G}x_{ij}=\bar c_i(x)-\delta\ge 0.
\]
Part~(2) holds because neither the price vector nor the spending vector changes.

For part~(3), the compressed equilibrium identities give, for every buyer $k\in B$,
\[
\sum_{j\in G} z_{kj}+\rho_k = e_k-r_k'.
\]
Hence
\[
\sum_{j\in G} z_{kj}+\widehat\rho_k
=
\sum_{j\in G} z_{kj}+\rho_k+r_k'
=
e_k,
\]
so the buyer-budget identities of the original instance hold.
Market clearing and the support condition are unchanged.

It remains to verify refund complementarity for $\widehat\rho$.
If $\widehat\rho_k>0$, then either $\rho_k>0$, in which case
$\alpha_k(q)\le 1$ by refund complementarity in the compressed equilibrium, or
$r_k'>0$.
In the latter case, some positive amount of refund had already been committed to
buyer $k$ at an earlier state where $\alpha_k=1$.
Since prices only increase afterwards and $q\ge p$, \Cref{lem:alpha-monotone-strong}
implies $\alpha_k(q)\le 1$.
Conversely, if $\alpha_k(q)>1$, then $\rho_k=0$ by refund complementarity in the
compressed equilibrium, and also $r_k'=0$ by the same monotonicity argument.
Thus $\widehat\rho_k=0$.

Therefore $(q,z,\widehat\rho)$ is an Arctic equilibrium of the original instance.
\end{proof}

\begin{lemma}[Effect of a commitment on component surplus]
\label{lem:refund-changes-surplus}
If a safe commitment of size $\delta$ is made to buyer $i$ in a compressed state $(p,x,r)$, then $s(p,r',H)=s(p,r,H)-\delta$ if $i\in H\cap B$, and $s(p,r',H)=s(p,r,H)$ otherwise.
\end{lemma}

\begin{proof}
Only $\bar e_i$ changes, and it decreases by exactly $\delta$.
\end{proof}

\subsection{Abundant edges and fertile components}
\label{subsec:abundant-fertile}

The abundance notion is unchanged from the weakly polynomial analysis.
Fertility is measured with effective budgets so that it reflects the compressed state.

\begin{definition}[$\Delta$-components]
\label{def:delta-components-strong}
Let $(p,x,r)$ be a compressed Arctic state at the beginning of a $\Delta$-scaling phase.
Recall that an edge $(i,j)\in B\times G$ is $\Delta$-abundant if
\[
x_{ij}\ge 3n\Delta.
\]
Write
\[
\mathcal A(\Delta):=\{(i,j)\in B\times G : x_{ij}\ge 3n\Delta\}.
\]
The connected components of the bipartite graph $(B\cup G,\mathcal A(\Delta))$ are called the
\emph{$\Delta$-components}, and their collection is denoted by $\mathcal C(\Delta)$.
\end{definition}

\begin{definition}[$\Delta$-residual network]
\label{def:delta-residual-strong}
The \emph{$\Delta$-residual network} $N_\Delta(p)$ is the directed graph on $B\cup G$ with a forward arc for every equality edge $(i,j)\in E(p)$ and a backward arc for every abundant edge $(i,j)\in\mathcal A(\Delta)$.
For a set $H\subseteq B\cup G$ we write
\[
\operatorname{ActiveSet}(p,\Delta,H):=\bigl\{k\in B\cup G:\exists v\in H\text{ with a directed path from } v \text{ to } k \text{ in } N_\Delta(p)\bigr\}.
\]
\end{definition}

\begin{definition}[$\Delta$-fertile component]
\label{def:delta-fertile-strong}
Let $(p,x,r)$ be a compressed Arctic state at the beginning of a $\Delta$-scaling phase and let $H\in\mathcal C(\Delta)$.
The component $H$ is \emph{$\Delta$-fertile} if one of the following holds.
\begin{enumerate}
    \item $H=\{i\}$ for some buyer $i\in B$ with $\alpha_i(p)>1$ and $\bar c_i(x)>\Delta/(3n^2)$.
    \item $s(p,r,H)\le -\Delta/(3n^2)$.
\end{enumerate}
The compressed state $(p,x,r)$ is \emph{$\Delta$-fertile} if some $\Delta$-component is $\Delta$-fertile.
\end{definition}

The Arctic singleton-buyer case differs from the linear Fisher one only in the extra assumption $\alpha_i(p)>1$: a buyer with $\alpha_i(p)\le 1$ does not force future support edges, because her optimal behavior is to keep her remaining cash as refund.

We record two invariants that the algorithm will preserve.

\begin{lemma}[Persistence of abundance]
\label{lem:abundance-persists}
If an edge is $\Delta$-abundant at the beginning of a $\Delta$-scaling phase, then it is $(\Delta/2)$-abundant at the beginning of the next phase.
In particular, once an edge becomes abundant, it remains abundant in all subsequent phases.
\end{lemma}

\begin{proof}
By \Cref{lem:per-phase-change}, a fixed spending variable changes during the phase by at most $n\Delta$, and the halving repair of \Cref{prop:halving} changes it by at most $\Delta/2$.
Commitments do not modify the spending vector.
Hence $x_{ij}\ge 3n\Delta$ at the beginning of the phase implies
\[
x_{ij}\ge 3n\Delta-n\Delta-\frac{\Delta}{2}>\frac{3n\Delta}{2}=3n\!\left(\frac{\Delta}{2}\right)
\]
at the beginning of the next phase.
The second claim follows by induction.
\end{proof}

\begin{lemma}[Fertile component yields progress]
\label{lem:fertile-gives-abundant}
Let $(p,x,r)$ be a compressed Arctic state at the beginning of a $\Delta$-scaling phase, and suppose some component $H\in\mathcal C(\Delta)$ is $\Delta$-fertile.
Within at most $5\log n+5$ subsequent ordinary scaling phases, at least one of the following occurs.
\begin{enumerate}
    \item A new abundant edge appears in the spending graph.
    \item $H=\{i\}$ is a singleton buyer component and $\alpha_i(\tilde p)\le 1$ at some intermediate price vector $\tilde p$; in particular, $H$ permanently loses fertility condition~(1) of \Cref{def:delta-fertile-strong}.
\end{enumerate}
\end{lemma}

\begin{proof}
Let $\Delta'$ denote the scaling parameter after $5\log n+5$ further phases, so that $\Delta'<\Delta/(16n^5)$, and let $p'$ denote the price vector at that moment.

\emph{Non-singleton case.}
Assume $s(p,r,H)\le -\Delta/(3n^2)$.
Prices are nondecreasing throughout ordinary phases, and effective budgets are nonincreasing (committed refunds can only grow), so
\[
s(p',r',H)\le s(p,r,H)\le -\frac{\Delta}{3n^2}<-3n^2\Delta',
\]
using $\Delta'<\Delta/(16n^5)<\Delta/(9n^4)$.
The $\Delta'$-feasibility conditions give $b_j(p',x')\in [0,\Delta']$ for every good $j$, and $\bar c_i(x')\ge 0$ for every buyer.
Combining the buyer-budget identity $\sum_{j\in G}x'_{ij}=\bar e_i-\bar c_i(x')$ with market clearing $\sum_{i\in B}x'_{ij}=p'_j+b_j(p',x')$ for $j\in H\cap G$ and subtracting gives the identity
\[
\underbrace{\sum_{i\notin H\cap B}\sum_{j\in H\cap G}x'_{ij}}_{\text{inflow}}-\underbrace{\sum_{i\in H\cap B}\sum_{j\notin H\cap G}x'_{ij}}_{\text{outflow}}
=-s(p',r',H)+\sum_{j\in H\cap G}b_j(p',x')+\sum_{i\in H\cap B}\bar c_i(x').
\]
Since the right-hand side is at least $-s(p',r',H)>3n^2\Delta'$, and outflow $\ge 0$, the inflow is also at least $3n^2\Delta'$.

The perturbation from \Cref{sec:prelim} ensures that the equality graph $E(p')$ is cycle-free, and hence the support of $x'$, being contained in $E(p')$, is cycle-free as well.
Therefore the number of positive-spending cut edges from $B\setminus H$ to $H\cap G$ is at most $n-1$.
Hence some single cut edge $(i,j)$ with $i\notin H\cap B$, $j\in H\cap G$ carries more than
\[
\frac{3n^2\Delta'}{n-1}\;=\;3n\Delta'+\frac{3n\Delta'}{n-1}\;>\;3n\Delta'
\]
units of spending, producing a new $\Delta'$-abundant edge.
Outcome~(1) holds.

\emph{Singleton case.}
Now assume $H=\{i\}$ with $\alpha_i(p)>1$ and $\bar c_i(x)>\Delta/(3n^2)$.
By \Cref{lem:alpha-monotone-strong}, $\alpha_i$ is nonincreasing along the trajectory of price vectors.
If $\alpha_i(\tilde p)\le 1$ at some intermediate state $\tilde p$ within the $5\log n+5$ phases, outcome~(2) holds directly (and by \Cref{lem:alpha-monotone-strong} the drop is permanent).

Otherwise $\alpha_i(p^{(t)})>1$ at every phase $t\in[k,k+5\log n+5]$, where $k$ denotes the current phase.
Commitments occur only at buyers with bang-per-buck exactly one, so buyer $i$'s committed refund $r_i$ is unchanged throughout these phases; in particular $\bar e_i$ is constant and equal to its value at phase~$k$.
At the end of the final phase, $i\in B_{>1}(p')$, so by $\Delta'$-optimality (\Cref{def:delta-optimal-arctic}) we have $\bar c_i(x')<\Delta'$.
The total spending of buyer $i$ therefore satisfies
\[
\sum_{j\in G}x'_{ij}=\bar e_i-\bar c_i(x')>\bar e_i-\Delta'\ge \bar c_i(x)-\Delta'>\frac{\Delta}{3n^2}-\Delta'>3n^2\Delta',
\]
where the inequality $\bar e_i\ge \bar c_i(x)$ is immediate from $\bar e_i=\bar c_i(x)+\sum_j x_{ij}$.
Since this total is spread over at most $|G|<n$ goods, at least one good $j$ satisfies $x'_{ij}>3n\Delta'$, i.e.\ a new $\Delta'$-abundant edge incident to $i$ appears.
Outcome~(1) holds.
\end{proof}

\subsection{The algorithm and the main theorem}
\label{subsec:algorithm-main}

With the abundance and fertility notions in place, we can state the strongly polynomial scaling algorithm at a high level.
It alternates ordinary $\Delta$-scaling phases (implemented by \textsc{PriceAndAugmentWithCriticalBuyer}) with invocations of a restart subroutine \textsc{MakeFertileAA}, whose role is to jump directly to a much smaller scaling parameter whenever an ordinary phase fails to make progress.

\begin{algorithm}[H]
\caption{\textsc{StrongScalingAA}$(e,U)$}
\label{alg:strong-scaling-aa}
\begin{algorithmic}[1]
\State Compute the initial parameter $\Delta_0$, prices $p^0$, and allocation $x^0$ as in \Cref{sec:modified-price-augment}; set $r^0\leftarrow 0$
\State $\Delta\leftarrow \Delta_0$, $(p,x,r)\leftarrow (p^0,x^0,r^0)$, $\mathrm{Threshold}\leftarrow \Delta$
\While{$\operatorname{BasicSolution}_{e-r}(\mathcal A(\Delta))$ is not an Arctic equilibrium of the compressed instance with budget vector $e-r$}
    \While{$(p,x,r)$ is not $\Delta$-optimal}
        \If{there exists a buyer $i\in B$ with $\alpha_i(p)\le 1$ and $\bar c_i(x)\ge \Delta$}
            \State $r_i\leftarrow r_i+\Delta$
        \Else
            \State $(p,x,r)\leftarrow \textsc{PriceAndAugmentWithCriticalBuyer}(p,x,r,\Delta)$
        \EndIf
    \EndWhile
    \If{$(p,x,r)$ is not $\Delta$-fertile \textbf{and} $\Delta\le \mathrm{Threshold}$}
        \State $(\Delta,p,x,r,\mathrm{Threshold})\leftarrow \textsc{MakeFertileAA}(p,x,r,\Delta)$
    \Else
        \State Replace $(p,x,r)$ by $(p,\widetilde x,r)$ from \Cref{prop:halving}; $\Delta\leftarrow \Delta/2$
    \EndIf
\EndWhile
\State $(p^\star,x^\star,\rho^\star)\leftarrow \operatorname{BasicSolution}_{e-r}(\mathcal A(\Delta))$
\State \Return $(p^\star,x^\star,\rho^\star+r)$
\end{algorithmic}
\end{algorithm}
The key to strong polynomiality is the subroutine \textsc{MakeFertileAA}, which is called whenever the state is not fertile and ordinary halving alone would not guarantee that a new abundant edge appears soon.
Specifically, \textsc{MakeFertileAA} takes a $\Delta$-optimal non-fertile compressed state and returns either (a) a ``delayed-discovery'' signal that simply lowers the threshold and promises a new abundant edge within $O(\log n)$ ordinary phases, or (b) a compressed restart to a much smaller parameter $\Delta'\le \Delta/n^2$ with the same promise.
We defer the construction of \textsc{MakeFertileAA} to \Cref{subsec:makefertile-overview} and the subsequent subsections; the precise guarantee is the following.

\begin{theorem}[Restart lemma]
\label{thm:restart}
There is a procedure $\textsc{MakeFertileAA}(p,x,r,\Delta)$ that, when applied to a
$\Delta$-optimal non-fertile compressed Arctic state, returns a tuple
\[
(\Delta',p',x',r',\mathrm{Threshold}')
\]
with one of the following outcomes.
\begin{itemize}
    \item \emph{Delayed discovery.}
    Here $\Delta'=\Delta$ and $(p',x',r')=(p,x,r)$.
    The procedure sets
    \[
    \mathrm{Threshold}'=\Delta/n^5
    \]
    and guarantees that a new abundant edge appears within $O(\log n)$ ordinary
    scaling phases.

    \item \emph{Compressed restart.}
    Here
    \[
    0<\Delta'\le \Delta/n^2.
    \]
    The procedure returns a new compressed state $(p',x',r')$, sets
    \[
    \mathrm{Threshold}'=\Delta'/n^5,
    \]
    preserves all previously abundant edges, ensures that $(p',x',r')$ becomes
    $\Delta'$-feasible within $O(\log n)$ ordinary phases, and guarantees that
    within a further $O(\log n)$ ordinary phases a new abundant edge appears.
\end{itemize}
Each call of \textsc{MakeFertileAA} runs in $O\bigl(n^2(m+n\log n)\bigr)$ time.
\end{theorem}

Assuming \Cref{thm:restart}, the main theorem follows from a direct counting argument.

\begin{theorem}[Strongly polynomial algorithm for the Arctic Auction]
\label{thm:main-strong}
\Cref{alg:strong-scaling-aa} computes an Arctic equilibrium of the original instance in $O(n\log n)$ scaling phases, with $O(n)$ calls to \textsc{MakeFertileAA}.
Its total running time is
\[
O\!\left(n^3(m+n\log n)\right).
\]
\end{theorem}

\begin{proof}
By \Cref{lem:abundance-persists}, once an edge becomes abundant it remains abundant forever.
The perturbation assumption makes the equilibrium support a forest with at most $n-1$ edges, so at most $n-1$ distinct abundant edges are ever discovered.

Define a \emph{progress event} to be either (A)~the discovery of a new abundant edge, or (B)~the first phase at which some singleton buyer $\{i\}$ has $\alpha_i(p)\le 1$.
Event~(A) occurs at most $n-1$ times; event~(B) occurs at most $|B|\le n$ times, because once a buyer has $\alpha_i\le 1$ it remains so forever by \Cref{lem:alpha-monotone-strong}.
Hence progress events are bounded in total by $2n-1=O(n)$.

We show that consecutive progress events are separated by $O(\log n)$ phases.
Consider any phase $k$ at which no progress event has yet occurred since the previous one.
If the state is $\Delta$-fertile at the start of phase $k$, then by \Cref{lem:fertile-gives-abundant} within $5\log n+5$ further phases either a new abundant edge appears (event~A) or some singleton buyer has $\alpha_i\le 1$ (event~B).
If the state is not $\Delta$-fertile and $\Delta\le \mathrm{Threshold}$, then \Cref{alg:strong-scaling-aa} invokes \textsc{MakeFertileAA}, and by \Cref{thm:restart} a new abundant edge appears within $O(\log n)$ further phases.
If the state is not $\Delta$-fertile and $\Delta>\mathrm{Threshold}$, ordinary halving continues until $\Delta$ drops below the threshold, which takes $O(\log n)$ phases (because the threshold was lowered to the current $\Delta/n^5$ at the previous restart, so at most $5\log n$ halvings bring $\Delta$ below the threshold).
In all cases, a progress event occurs within $O(\log n)$ phases.

Therefore the total number of scaling phases is $O(n\log n)$.
Each ordinary $\Delta$-phase performs at most $n$ calls of \textsc{PriceAndAugmentWithCriticalBuyer}, each implementable in $O(m+n\log n)$ time by the Fibonacci-heap analysis of \Cref{sec:modified-price-augment}, for $O(n(m+n\log n))$ time per phase.
By \Cref{thm:restart} each \textsc{MakeFertileAA} call also costs $O(n^2(m+n\log n))$; there are at most $O(n)$ such calls because they are separated from each other by $\Omega(\log n)$ phases.
Combining, the total running time is $O(n^3(m+n\log n))$.

Upon termination, let
\[
(p^\star,x^\star,\rho^\star)=\operatorname{BasicSolution}_{e-r}(\mathcal A(\Delta)).
\]
By repeated application of \Cref{lem:safe-refund-commitment}(3), the triple
\[
(p^\star,x^\star,\rho^\star+r)
\]
is an Arctic equilibrium of the original instance.
\end{proof}

\begin{corollary}
\label{cor:main-strong}
The Arctic Auction with linear utilities admits a strongly polynomial algorithm.
\end{corollary}

\begin{proof}
Immediate from \Cref{thm:main-strong}.
\end{proof}

\subsection{The \texorpdfstring{\textsc{MakeFertileAA}}{MakeFertileAA} subroutine}
\label{subsec:makefertile-overview}

The remainder of \Cref{sec:strongly-poly} is devoted to the construction and analysis of \textsc{MakeFertileAA}, culminating in the proof of \Cref{thm:restart}.
The subroutine is built from three auxiliary procedures.

\begin{itemize}
    \item \textsc{GetParameterAA} (\Cref{alg:getparameter-aa}) computes the next scaling parameter $\Delta'$.singleton good component.
    \item \textsc{GetPricesAA} (\Cref{alg:getprices-aa}) computes the next price vector $p'$ and the next refund vector $r'$ by aggregating prices returned by \textsc{SpecialPriceAA} across components.
    \item \textsc{GetAllocationsAA} (\Cref{alg:getallocations-aa}) constructs a new allocation $x'$ supported only on abundant edges.
\end{itemize}

The main subroutine is assembled from these pieces as follows.

\begin{algorithm}[H]
\caption{\textsc{MakeFertileAA}$(p,x,r,\Delta)$}
\label{alg:makefertile-aa}
\begin{algorithmic}[1]
\State $\Delta'\leftarrow\textsc{GetParameterAA}(p,x,r,\Delta)$
\If{$\Delta'>\Delta/n^2$}
    \State \Return $(\Delta,p,x,r,\Delta/n^5)$
\Else
    \State $(p',r')\leftarrow\textsc{GetPricesAA}(p,x,r,\Delta')$
    \State $x'\leftarrow\textsc{GetAllocationsAA}(p',r',\Delta')$
    \State \Return $(\Delta',p',x',r',\Delta'/n^5)$
\EndIf
\end{algorithmic}
\end{algorithm}

The correctness of \textsc{MakeFertileAA} is expressed in \Cref{thm:restart}, whose proof is deferred to \Cref{subsec:restart-proof} after all three auxiliary procedures have been introduced.
The remaining subsections are organized as follows: \Cref{subsec:specialprice} introduces the price-raising subroutine \textsc{SpecialPriceAA}; \Cref{subsec:auxnet} develops the auxiliary network and multiplier-path machinery needed to reason about \textsc{SpecialPriceAA}; \Cref{subsec:getparam} gives \textsc{GetParameterAA} and \textsc{GetPricesAA} together with the technical restart invariants; \Cref{subsec:getallocations} gives \textsc{GetAllocationsAA}; and \Cref{subsec:restart-proof} completes the proof of \Cref{thm:restart}.

\subsection{The price-raising procedure \texorpdfstring{\textsc{SpecialPriceAA}}{SpecialPriceAA}}
\label{subsec:specialprice}

The restart uses a price-raising routine that, given a $\Delta$-component $H$ and a target $t\ge 0$, raises prices on the active set of $H$ until the effective surplus of $H$ reaches $t$ or some other component becomes sufficiently negative.
In the Arctic setting we add one extra stopping event: if an active buyer with positive effective cash becomes critical, we commit as much of her remaining cash to refund as is safe; once her effective cash is exhausted, the event no longer fires at her and price-raising proceeds past $\alpha_i=1$ without further interruption from that buyer.

During \textsc{SpecialPriceAA} the spending vector $x$ is fixed while the refund
vector changes. We therefore write
\[
\bar c_i^{\,\hat r}(x):=e_i-\hat r_i-\sum_{j\in G}x_{ij}
\]
for the current effective cash of buyer $i$ inside the subroutine.

\begin{algorithm}[H]
\caption{\textsc{SpecialPriceAA}$(p,x,r,H,t)$}
\label{alg:special-price-aa}
\begin{algorithmic}[1]
\Require A compressed Arctic state $(p,x,r)$ at the end of a $\Delta$-phase, a $\Delta$-component $H\in\mathcal C(\Delta)$, and a target $t\ge 0$
\State $\hat p\leftarrow p$, $\hat r\leftarrow r$
\While{$s(\hat p,\hat r,H)>t$ \textbf{and} $s(\hat p,\hat r,J)>-s(\hat p,\hat r,H)/(2n^2)$ for every $J\in\mathcal C(\Delta)$}
    \State $A\leftarrow\operatorname{ActiveSet}(\hat p,\Delta,H)$
    \State Let $q\ge 1$ be the \emph{smallest} multiplier such that scaling $\hat p_j$ by $q$ for every $j\in A\cap G$ causes at least one of the following events to first occur:
    \begin{enumerate}
        \item a new equality edge $(i,j)\in E(\hat p)\setminus E(p)$ appears with $i$ inactive and $j$ active;
        \item $s(\hat p,\hat r,H)=t$;
        \item $s(\hat p,\hat r,J)=-s(\hat p,\hat r,H)/(2n^2)$ for some $J\in\mathcal C(\Delta)$;
        \item some active buyer $i\in A\cap B$ satisfies $\alpha_i(\hat p)=1$ and $\bar c_i^{\,\hat r}(x)>0$.
    \end{enumerate}
    \State Apply the multiplier: set $\hat p_j\leftarrow q\,\hat p_j$ for every $j\in A\cap G$
    \If{event~(4) occurs at some buyer $i$}
        \If{$i\in H$}
            \State $\delta\leftarrow \min\{\bar c_i^{\,\hat r}(x),\, s(\hat p,\hat r,H)-t\}$
        \Else
            \State Let $J(i)\in\mathcal C(\Delta)$ be the unique component with $i\in J(i)$
            \State $\delta\leftarrow \min\{\bar c_i^{\,\hat r}(x),\, s(\hat p,\hat r,J(i))+s(\hat p,\hat r,H)/(2n^2)\}$
        \EndIf
        \State $\hat r_i\leftarrow \hat r_i+\delta$
    \EndIf
\EndWhile
\State \Return $(\hat p,\hat r)$
\end{algorithmic}
\end{algorithm}

\begin{lemma}[Monotonicity inside \textsc{SpecialPriceAA}]
\label{lem:specialprice-monotone}
During \Cref{alg:special-price-aa}, $\hat p$ is nondecreasing and $s(\hat p,\hat r,H)$ is nonincreasing.
The equality graph $E(\hat p)$ only gains edges, and $\mathcal A(\Delta)$ is unchanged.
\end{lemma}

\begin{proof}
Only prices of active goods are multiplied, so $\hat p$ is nondecreasing
coordinatewise; commitments do not change prices.
The spending vector $x$ is fixed throughout the subroutine, hence
$\mathcal A(\Delta)$ is unchanged.

We next prove that equality edges cannot disappear.
At any moment, an equality edge cannot join an active buyer to an inactive good,
because that good would then be reachable by a forward arc and hence active.
Likewise, an equality edge cannot join an inactive buyer to an active good, because
such an edge would itself be a newly appearing active/inactive edge.
Therefore every existing equality edge is either entirely inside the active set or
entirely outside it.
Edges outside the active set are unaffected by the price update.
For an active buyer, all of her equality goods are active and are all scaled by the
same factor, so every pre-existing active equality edge remains an equality edge.
Hence $E(\hat p)$ can only gain edges; it never loses one.

Finally, $H\subseteq \operatorname{ActiveSet}(\hat p,\Delta,H)$ throughout the
procedure, so every price increase on a good of $H\cap G$ decreases
$s(\hat p,\hat r,H)=\bar e_{H\cap B}-\hat p_{H\cap G}$.
A commitment at a buyer of $H$ also decreases the same surplus by
\Cref{lem:refund-changes-surplus}, while a commitment outside $H$ leaves it
unchanged.
Thus $s(\hat p,\hat r,H)$ is nonincreasing.
\end{proof}

\begin{lemma}[Termination of \textsc{SpecialPriceAA}]
\label{lem:specialprice-terminates}
\Cref{alg:special-price-aa} terminates in at most $n+|B|$ iterations of the outer while-loop.
\end{lemma}

\begin{proof}
Each iteration ends when one of the events~(1)--(4) first becomes tight.
Events~(2) and~(3) immediately terminate the while-loop.

By \Cref{lem:specialprice-monotone}, the equality graph only gains edges, so
event~(1) can occur at most $n-1$ times over the whole execution.

Consider event~(4) at a buyer $i$.
By the definition of $\delta$, either:
\begin{enumerate}
    \item $\delta=\bar c_i^{\,\hat r}(x)$, in which case buyer $i$ has zero effective
    cash afterwards and event~(4) can never fire again at $i$; or
    \item $\delta$ makes one of the two while-inequalities tight, in which case the
    loop exits at the next test.
\end{enumerate}
Hence event~(4) contributes at most one nonterminal iteration per buyer, for a
total of at most $|B|$ such iterations.

Therefore the total number of outer iterations is at most $n+|B|$.
\end{proof}

\begin{lemma}[Output of \textsc{SpecialPriceAA}]
\label{lem:specialprice-output}
At termination, $(\hat p,\hat r)$ satisfies exactly one of:
\begin{enumerate}
    \item $s(\hat p,\hat r,H)=t$ and $s(\hat p,\hat r,J)\ge -t/(2n^2)$ for every $J\in\mathcal C(\Delta)$, or
    \item there exists $J\in\mathcal C(\Delta)$ with $s(\hat p,\hat r,J)=-s(\hat p,\hat r,H)/(2n^2)$ and $s(\hat p,\hat r,H)\ge t$.
\end{enumerate}
\end{lemma}

\begin{proof}
The while-loop exits only when one of the two inequalities becomes tight.
By the choice of $\delta$, each commitment preserves the other inequality.
Therefore, if the first becomes tight, we are in case~(1); if the second becomes tight, we are in case~(2).
\end{proof}

\subsection{The auxiliary network and maximum multiplier paths}
\label{subsec:auxnet}

We now introduce a weighted graph that encapsulates the price dynamics of $\textsc{SpecialPriceAA}$ independently of the effective budgets.

\begin{definition}[Auxiliary network]
\label{def:auxnet}
The \emph{auxiliary network} $N^*$ is the directed graph on $B\cup G$ with:
\begin{itemize}
    \item a forward arc $(i,j)$ of weight $U_{ij}$ for every $(i,j)\in B\times G$ with $U_{ij}>0$;
    \item a backward arc $(j,i)$ of weight $1/U_{ij}$ for every abundant edge $(i,j)\in \mathcal A(\Delta)$.
\end{itemize}
For a directed path $P$, its \emph{utility} is the product of the arc weights along $P$, denoted by $U(P)$.
For $v,w\in B\cup G$, define
\[
\mu(v,w):=\max\{U(P): P \text{ is a directed path from } v \text{ to } w\},
\]
with $\mu(v,w):=0$ if no such path exists.

For each $\Delta$-component $H$ with $H\cap G\neq \emptyset$, fix a root $v(H)\in H\cap G$.
If $H$ is a singleton buyer component, no root is needed in this subsection.
For components $H,K$ with $H\cap G\neq \emptyset$ and $K\cap G\neq \emptyset$, define
\[
\mu(H,K):=\mu(v(H),v(K)).
\]
\end{definition}

\begin{lemma}[Cycle-utility bound]
\label{lem:cycle-utility}
Every directed cycle $C$ in $N^*$ has $U(C)\le 1$.
\end{lemma}

\begin{proof}
For each forward arc $(i,j)$ of $N^*$, define the rescaled weight
\[
\widetilde U(i,j):=\frac{U_{ij}}{\alpha_i(p)\,p_j}.
\]
For each backward arc $(j,i)$ corresponding to an abundant edge $(i,j)\in \mathcal A(\Delta)$, define
\[
\widetilde U(j,i):=\frac{\alpha_i(p)\,p_j}{U_{ij}}.
\]
Since the current state is $\Delta$-feasible and every abundant edge has positive spending, every abundant edge is an equality edge at the current price vector $p$.
Hence for every backward arc $(j,i)$ we have
\[
\widetilde U(j,i)=1.
\]
Also, for every forward arc $(i,j)$,
\[
\widetilde U(i,j)=\frac{U_{ij}}{\alpha_i(p)\,p_j}\le 1
\]
by the definition of $\alpha_i(p)$.

Now let $C$ be a directed cycle in $N^*$.
Multiplying the rescaled weights around $C$, all factors of the form $\alpha_i(p)$ and $p_j$ cancel, because each buyer and each good appears equally often as a head and as a tail along the cycle.
Therefore
\[
\widetilde U(C)=U(C).
\]
Since every factor in $\widetilde U(C)$ is at most $1$, it follows that
\[
U(C)=\widetilde U(C)\le 1.
\]
\end{proof}

\begin{lemma}[Abundant edges remain equality throughout \textsc{SpecialPriceAA}]
\label{lem:abundant-stays-equality}
During the execution of \Cref{alg:special-price-aa}, every abundant edge $(i,j)\in \mathcal A(\Delta)$ remains an equality edge at the current price vector $\hat p$; that is,
\[
\frac{U_{ij}}{\hat p_j}=\alpha_i(\hat p).
\]
\end{lemma}

\begin{proof}
Fix an abundant edge $(i,j)\in \mathcal A(\Delta)$.

If $j$ is inactive throughout the procedure, then $\hat p_j=p_j$.
Since $(i,j)$ is an equality edge at the initial price vector $p$, we have
\[
\frac{U_{ij}}{\hat p_j}=\frac{U_{ij}}{p_j}=\alpha_i(p).
\]
Also, by \Cref{lem:alpha-monotone-strong}, $\alpha_i(\hat p)\le \alpha_i(p)$.
On the other hand,
\[
\alpha_i(\hat p)\ge \frac{U_{ij}}{\hat p_j}=\alpha_i(p),
\]
so equality must hold:
\[
\alpha_i(\hat p)=\frac{U_{ij}}{\hat p_j}.
\]

Now suppose that $j$ becomes active during the procedure.
Because $(j,i)$ is a backward arc of $N_\Delta(\hat p)$, once $j$ is active the buyer $i$ also becomes active.
From that point onward, every equality good of buyer $i$ is active as well, and all active goods are scaled by the same cumulative factor, say $q$.
Hence
\[
\alpha_i(\hat p)=\frac{\alpha_i(p)}{q},
\qquad
\hat p_j=q\,p_j.
\]
Since $(i,j)$ was an equality edge at $p$, we obtain
\[
\frac{U_{ij}}{\hat p_j}
=
\frac{U_{ij}}{q\,p_j}
=
\frac{\alpha_i(p)}{q}
=
\alpha_i(\hat p).
\]
Thus $(i,j)$ remains an equality edge throughout the procedure.
\end{proof}

\begin{lemma}[Active set and max-multiplier identity]
\label{lem:active-multiplier}
Let $K,H\in \mathcal C(\Delta)$ be components with $K\cap G\neq \emptyset$ and $H\cap G\neq \emptyset$.
Let
\[
(\hat p,\hat r)=\textsc{SpecialPriceAA}(p,x,r,K,t),
\]
and suppose that
\[
H\subseteq \operatorname{ActiveSet}(\hat p,\Delta,K)
\]
at termination.
Then
\[
\mu(K,H)=\frac{\hat p_{v(H)}}{\hat p_{v(K)}}.
\]
\end{lemma}

\begin{proof}
For each forward arc $(i,j)$ of $N^*$, define
\[
\widetilde U(i,j):=\frac{U_{ij}}{\alpha_i(\hat p)\,\hat p_j},
\]
and for each backward arc $(j,i)$ corresponding to an abundant edge $(i,j)\in \mathcal A(\Delta)$, define
\[
\widetilde U(j,i):=\frac{\alpha_i(\hat p)\,\hat p_j}{U_{ij}}.
\]
Every forward arc satisfies $\widetilde U(i,j)\le 1$, and by \Cref{lem:abundant-stays-equality}, every backward arc satisfies $\widetilde U(j,i)=1$.

As in the proof of \Cref{lem:cycle-utility}, for every directed cycle $C$ in $N^*$ the rescaled product telescopes, so
\[
\widetilde U(C)=U(C)\le 1.
\]

Now let $P$ be any directed path in $N^*$ from $v(K)$ to $v(H)$.
Again by telescoping,
\[
\widetilde U(P)=U(P)\cdot \frac{\hat p_{v(K)}}{\hat p_{v(H)}}.
\]
Since every factor in $\widetilde U(P)$ is at most $1$ on forward arcs and equal to $1$ on backward arcs, we obtain
\[
U(P)\le \frac{\hat p_{v(H)}}{\hat p_{v(K)}}.
\]
Taking the maximum over all such paths yields
\[
\mu(K,H)\le \frac{\hat p_{v(H)}}{\hat p_{v(K)}}.
\]

On the other hand, because $H\subseteq \operatorname{ActiveSet}(\hat p,\Delta,K)$, there is a directed path from $v(K)$ to $v(H)$ in $N_\Delta(\hat p)$.
Every forward arc on this path is an equality edge at $\hat p$, and every backward arc corresponds to an abundant edge, which is also an equality edge at $\hat p$ by \Cref{lem:abundant-stays-equality}.
Therefore every arc on this path has rescaled weight exactly $1$, and so for this path $P$,
\[
U(P)=\frac{\hat p_{v(H)}}{\hat p_{v(K)}}.
\]
Hence
\[
\mu(K,H)\ge \frac{\hat p_{v(H)}}{\hat p_{v(K)}}.
\]
Combining the two inequalities gives the claim.
\end{proof}

\begin{lemma}[Domination]
\label{lem:domination}
Let $H,K\in \mathcal C(\Delta)$ be components with $H\cap G\neq \emptyset$ and $K\cap G\neq \emptyset$.
Let $p^H,p^K$ be the price vectors returned by
\[
\textsc{SpecialPriceAA}(p,x,r,H,t)
\qquad\text{and}\qquad
\textsc{SpecialPriceAA}(p,x,r,K,t),
\]
respectively.
If
\[
p^K_{v(H)}>p^H_{v(H)},
\]
then
\[
H\subseteq \operatorname{ActiveSet}(p^K,\Delta,K),
\]
and consequently
\[
\mu(K,H)=\frac{p^K_{v(H)}}{p^K_{v(K)}}.
\]
\end{lemma}

\begin{proof}
If $v(H)$ were inactive throughout the execution of $\textsc{SpecialPriceAA}(p,x,r,K,t)$, then its price would remain unchanged, so
\[
p^K_{v(H)}=p_{v(H)}.
\]
On the other hand, prices are nondecreasing during every execution of \textsc{SpecialPriceAA}, hence
\[
p_{v(H)}\le p^H_{v(H)}.
\]
Therefore
\[
p^K_{v(H)}=p_{v(H)}\le p^H_{v(H)},
\]
contradicting the assumption.
Thus $v(H)$ must become active during the $K$-execution, which implies
\[
H\subseteq \operatorname{ActiveSet}(p^K,\Delta,K).
\]
The displayed identity now follows immediately from \Cref{lem:active-multiplier}.
\end{proof}

\subsection{Parameter and price selection}
\label{subsec:getparam}

For the remainder of \Cref{sec:strongly-poly} we assume that $(p,x,r)$ is a compressed Arctic state at the end of a $\Delta$-scaling phase, that $(p,x,r)$ is $\Delta$-optimal, and that $(p,x,r)$ is \emph{not} $\Delta$-fertile.
Let $\mathcal C:=\mathcal C(\Delta)$.

\begin{algorithm}[H]
\caption{\textsc{GetParameterAA}$(p,x,r,\Delta)$}
\label{alg:getparameter-aa}
\begin{algorithmic}[1]
\For{each $H\in\mathcal C$}
    \If{$H=\{i\}$ for some buyer $i\in B$}
    \If{$\alpha_i(p)>1$}
        \State $\Delta_H\leftarrow \bar c_i(x)$
    \Else
        \State $\Delta_H\leftarrow 0$
    \EndIf
    \Else
        \State $(\hat p,\hat r)\leftarrow
        \textsc{SpecialPriceAA}(p,x,r,H,0)$
        \State $\Delta_H\leftarrow s(\hat p,\hat r,H)$
    \EndIf
\EndFor
\State \Return $\Delta':=\max\{\Delta_H:H\in\mathcal C\}$
\end{algorithmic}
\end{algorithm}
\begin{algorithm}[H]
\caption{\textsc{GetPricesAA}$(p,x,r,\Delta')$}
\label{alg:getprices-aa}
\begin{algorithmic}[1]
\For{each $H\in\mathcal C$}
    \If{$|H|=1$}
        \State $(p^H,r^H)\leftarrow (p,r)$
        \Comment{Baseline run so that the maxima below are well defined.}
    \ElsIf{$s(p,r,H)\le \Delta'$}
        \State $(p^H,r^H)\leftarrow (p,r)$
    \Else
        \State $(p^H,r^H)\leftarrow
        \textsc{SpecialPriceAA}(p,x,r,H,\Delta')$
    \EndIf
\EndFor
\State For each $j\in G$, set $p'_j:=\max\{p^H_j:H\in\mathcal C\}$
\State For each $i\in B$, set $r'_i:=\max\{r^H_i:H\in\mathcal C\}$
\State \Return $(p',r')$
\end{algorithmic}
\end{algorithm}
\begin{proposition}
[Componentwise witness for \textsc{GetPricesAA} ]
\label{prop:component-witness}
For every non-singleton component $J\in\mathcal C$, there exists a component
$W(J)\in\mathcal C$ such that, if $(p^{W(J)},r^{W(J)})$ denotes the output of the
corresponding run used in \Cref{alg:getprices-aa}, then
\[
p'_j=p^{W(J)}_j \quad \forall j\in J\cap G,
\qquad
r'_i=r^{W(J)}_i \quad \forall i\in J\cap B.
\]
In other words, for each non-singleton component, the coordinatewise maxima used
to define $(p',r')$ are simultaneously realized by a single witness run on that
component.
\end{proposition}
\begin{proof}
Fix a non-singleton component $J\in\mathcal C$.
By the perturbation assumption from \Cref{subsec:generic-perturbation}, each
connected component of every equality graph contains at most one critical
buyer. Therefore, among all runs used in \Cref{alg:getprices-aa}, there is at
most one buyer of $J$ that can receive new refund. Denote this buyer by $i_J$
if it exists.

We first record a simple consequence of the definition of
\textsc{SpecialPriceAA}: during any fixed run, the goods of $J\cap G$ are
either all inactive or all active together, and whenever they are active they
are all multiplied by a common factor. Hence, for every run rooted at a
component $H\in\mathcal C$, there exists a scalar $q_H(J)\ge 1$ such that
\[
p_j^H=q_H(J)\,p_j \qquad \forall j\in J\cap G.
\]
In particular, a run maximizes $p_j^H$ for one good $j\in J\cap G$ if and only
if it maximizes $p_j^H$ for every good $j\in J\cap G$.

We now distinguish three cases according to the amount of refund committed in
$J$ by \Cref{alg:getprices-aa}.

\medskip
\noindent\emph{Case 1: no buyer of $J$ receives new refund.}
Equivalently, for every $i\in J\cap B$,
\[
r_i'=r_i.
\]
Choose a run $W(J)$ maximizing the common scaling factor on $J$, equivalently
maximizing $p_j^H$ for some (and therefore every) good $j\in J\cap G$. Then
\[
p_j^{W(J)}=p_j' \qquad \forall j\in J\cap G.
\]
Since no buyer of $J$ receives additional refund in any run,
\[
r_i^{W(J)}=r_i=r_i' \qquad \forall i\in J\cap B.
\]
Thus $W(J)$ is the required witness.

\medskip
\noindent\emph{Case 2: the unique refunding buyer $i_J$ satisfies}
\[
0<r'_{i_J}-r_{i_J}<\bar c_{i_J}(x).
\]
Let $H$ and $K$ be two runs with
\[
r_{i_J}^H>r_{i_J}
\qquad\text{and}\qquad
r_{i_J}^K>r_{i_J}.
\]
Then in both runs buyer $i_J$ becomes critical inside $J$, but the commitment
made at $i_J$ is strictly smaller than the full amount $\bar c_{i_J}(x)$.
At the moment when $i_J$ becomes critical, the common scaling factor on the
goods of $J$ is uniquely determined; therefore the terminal scaling factor on
$J$ is the same in both runs. Hence
\[
p_j^H=p_j^K \qquad \forall j\in J\cap G.
\]
So every run that gives a positive but non-full refund to $i_J$ induces the
same terminal prices on all goods of $J$.

Now suppose that some run $H'\in\mathcal C$ satisfies
\[
p_j^{H'}>p_j^H
\qquad\text{for some } j\in J\cap G.
\]
Since all goods of $J$ scale by the same factor in each run, this means that
the run rooted at $H'$ continues strictly past the critical scale at which
$i_J$ becomes critical inside $J$. But once the price-raising procedure moves
past that critical scale, buyer $i_J$ must receive the full remaining refund,
namely
\[
r_{i_J}^{H'}-r_{i_J}=\bar c_{i_J}(x),
\]
contradicting the assumption of the present case that the maximal additional
refund on $J$ is strictly smaller than $\bar c_{i_J}(x)$.

Therefore no run can produce a larger price on any good of $J$ than the common
value attained by the runs that refund $i_J$ partially. Choose any run $W(J)$
such that
\[
r_{i_J}^{W(J)}=r_{i_J}'.
\]
Then necessarily
\[
p_j^{W(J)}=p_j' \qquad \forall j\in J\cap G.
\]
Moreover, $i_J$ is the only buyer of $J$ receiving new refund, so for every
other buyer $\ell\in J\cap B\setminus\{i_J\}$,
\[
r_\ell^{W(J)}=r_\ell=r_\ell'.
\]
Hence
\[
r_i^{W(J)}=r_i' \qquad \forall i\in J\cap B,
\]
and $W(J)$ is again the required witness.

\medskip
\noindent\emph{Case 3: the unique refunding buyer $i_J$ satisfies}
\[
r'_{i_J}-r_{i_J}=\bar c_{i_J}(x).
\]
In this case some run commits the full remaining effective cash of $i_J$ to
refund. Consider any run in which $i_J$ becomes critical inside $J$. After the
critical scale is reached, every further increase of the prices on the goods of
$J$ requires passing beyond that critical point, and therefore forces the full
commitment of buyer $i_J$'s remaining effective cash. Hence every run whose
terminal scaling factor on $J$ is at least the critical scaling factor satisfies
\[
r_{i_J}^H-r_{i_J}=\bar c_{i_J}(x),
\]
that is,
\[
r_{i_J}^H=r_{i_J}'.
\]

Choose a run $W(J)$ maximizing the terminal scaling factor on $J$, equivalently
maximizing $p_j^H$ for one (and hence every) good $j\in J\cap G$. Then
\[
p_j^{W(J)}=p_j' \qquad \forall j\in J\cap G.
\]
By the previous paragraph, the same run also satisfies
\[
r_{i_J}^{W(J)}=r_{i_J}'.
\]
As before, no buyer of $J$ other than $i_J$ can receive new refund, so
\[
r_\ell^{W(J)}=r_\ell=r_\ell'
\qquad \forall \ell\in J\cap B\setminus\{i_J\}.
\]
Thus
\[
r_i^{W(J)}=r_i' \qquad \forall i\in J\cap B.
\]

In all three cases we have found a component $W(J)\in\mathcal C$ such that
\[
p_j^{W(J)}=p_j' \quad \forall j\in J\cap G,
\qquad
r_i^{W(J)}=r_i' \quad \forall i\in J\cap B.
\]
This proves the proposition.
\end{proof}

\begin{lemma}[Restart invariants]
\label{lem:restart-invariants}
Let $\Delta'$ and $(p',r')$ be the outputs of
\Cref{alg:getparameter-aa,alg:getprices-aa} applied to a non-fertile
$\Delta$-optimal compressed state $(p,x,r)$.
Then:
\begin{enumerate}
    \item for every $H\in\mathcal C$,
    \[
    s(p',r',H)\ge -\Delta'/n^2;
    \]
    \item for every $H\in\mathcal C$,
    \[
    s(p^H,r^H,H)=\min\{s(p,r,H),\Delta'\};
    \]
    \item either
    \begin{enumerate}
        \item[(a)] there exists a non-singleton $H^*\in\mathcal C$ such that
        \[
        p'_{v(H^*)}=p^{H^*}_{v(H^*)}
        \qquad\text{and}\qquad
        s(p',r',H^*)=\Delta',
        \]
        or
        \item[(b)] there exists a singleton buyer component $H^*=\{i^*\}$ such that
        \[
        \bar e_{i^*}=\Delta'.
        \]
    \end{enumerate}
    \item either
    \begin{enumerate}
        \item[(a)] there exists $J\in\mathcal C$ such that
        \[
        -\Delta'/n^2\le s(p',r',J)\le -\Delta'/(2n^2),
        \]
        or
        \item[(b)] there exists a singleton buyer component $H^*=\{i^*\}$ such that
        \[
        \bar e_{i^*}=\Delta'.
        \]
    \end{enumerate}
\end{enumerate}
\end{lemma}

\begin{proof}
For part~(1), we first note that by the perturbation assumption from
\Cref{subsec:generic-perturbation}, every connected component of every
equality graph contains at most one critical buyer. In particular, for each
fixed component $H\in\mathcal C$, at most one buyer of $H$ can receive new
refund during the runs used in \Cref{alg:getprices-aa}.

\medskip
\noindent\emph{Part~(1), singleton case.}
If $H=\{i\}$ for some buyer $i\in B$, then $H\cap G=\emptyset$, and hence
\[
s(p',r',H)=\bar e_i\ge 0>-\Delta'/n^2.
\]

\medskip
\noindent\emph{Part~(1), non-singleton case.}
Fix a non-singleton component $H\in\mathcal C$.

If no buyer in $H$ receives any additional refund during
\Cref{alg:getprices-aa}, then
\[
r_\ell'=r_\ell \qquad\text{for every }\ell\in H\cap B.
\]
Choose a component $K\in\mathcal C$ such that
\[
p_{v(H)}^K=p_{v(H)}'.
\]
Since the goods of $H\cap G$ are scaled uniformly within every single run,
this implies
\[
p_j^K=p_j' \qquad\text{for every }j\in H\cap G.
\]
Hence
\[
s(p',r',H)=s(p^K,r^K,H).
\]
If $K=H$, then
\[
s(p',r',H)=s(p^H,r^H,H)\ge 0,
\]
so
\[
s(p',r',H)\ge 0>-\Delta'/n^2.
\]
If $K\neq H$, then $H$ is a non-root component during the run rooted at $K$,
and by \Cref{lem:specialprice-output},
\[
s(p^K,r^K,H)\ge -\Delta_K/(2n^2)\ge -\Delta'/(2n^2)>-\Delta'/n^2.
\]
Thus part~(1) holds in this case.

Now suppose that some buyer of $H$ receives additional refund during
\Cref{alg:getprices-aa}. By the perturbation assumption, there is then a
unique such buyer; denote it by $i_H$.
Applying \Cref{prop:component-witness} to the pair $(i_H,v(H))$, we obtain a
component $K\in\mathcal C$ such that
\[
r_{i_H}^K=r_{i_H}'
\qquad\text{and}\qquad
p_{v(H)}^K=p_{v(H)}'.
\]
Again, since the goods of $H\cap G$ are scaled uniformly within every run,
\[
p_j^K=p_j' \qquad\text{for every }j\in H\cap G.
\]
Moreover, no buyer of $H$ other than $i_H$ can receive new refund, so
\[
r_\ell^K=r_\ell'=r_\ell
\qquad\text{for every }\ell\in H\cap B\setminus\{i_H\}.
\]
Therefore
\[
s(p',r',H)=s(p^K,r^K,H).
\]
If $K=H$, then
\[
s(p',r',H)=s(p^H,r^H,H)\ge 0.
\]
If $K\neq H$, then by \Cref{lem:specialprice-output},
\[
s(p^K,r^K,H)\ge -\Delta_K/(2n^2)\ge -\Delta'/(2n^2)>-\Delta'/n^2.
\]
This completes the proof of part~(1).

\medskip
\noindent\emph{Part~(2).}
If $H=\{i\}$ is a singleton buyer component, then by
\Cref{alg:getparameter-aa,alg:getprices-aa},
\[
\Delta_H=\bar e_i=s(p,r,H)
\qquad\text{and}\qquad
(p^H,r^H)=(p,r).
\]
Hence
\[
s(p^H,r^H,H)=\bar e_i=\min\{\bar e_i,\Delta'\}
           =\min\{s(p,r,H),\Delta'\}.
\]

If $H=\{j\}$ is a singleton good component, then
\textsc{SpecialPriceAA}$(p,x,r,H,0)$ returns immediately, so
\[
\Delta_H=s(p,r,H)=-p_j
\qquad\text{and}\qquad
(p^H,r^H)=(p,r).
\]
Hence
\[
s(p^H,r^H,H)=s(p,r,H)=\min\{s(p,r,H),\Delta'\}.
\]

Now let $H$ be non-singleton.
If $s(p,r,H)\le \Delta'$, then
\Cref{alg:getprices-aa} leaves $(p^H,r^H)=(p,r)$, so
\[
s(p^H,r^H,H)=s(p,r,H)=\min\{s(p,r,H),\Delta'\}.
\]

Assume instead that $s(p,r,H)>\Delta'$. Then
\textsc{SpecialPriceAA}$(p,x,r,H,\Delta')$ is invoked.
While
\[
s(\hat p,\hat r,H)>\Delta',
\]
the second while-inequality of \Cref{alg:special-price-aa} gives, for every
$J\in\mathcal C$,
\[
s(\hat p,\hat r,J)>
-\frac{s(\hat p,\hat r,H)}{2n^2}
>
-\frac{\Delta'}{2n^2}.
\]
Hence the lower-barrier event cannot become tight before the root surplus has
reached $\Delta'$. Therefore event~(2) must occur first, and so
\[
s(p^H,r^H,H)=\Delta'=\min\{s(p,r,H),\Delta'\}.
\]
This completes part~(2).

\medskip
\noindent\emph{Part~(3).}
Let
\[
D:=\{H\in\mathcal C:\ \Delta_H=\Delta'\}.
\]
The set $D$ is nonempty by definition of $\Delta'$.

If some $H\in D$ is a singleton buyer component, say $H=\{i\}$, then
\[
\bar e_i=\Delta_H=\Delta',
\]
and alternative~(b) holds.

Thus it remains to consider the case that every component of $D$ is
non-singleton. By part~(2), every $H\in D$ satisfies
\[
s(p^H,r^H,H)=\Delta'.
\]

We claim that there exists a component $H^*\in D$ such that
\[
p'_{v(H^*)}=p^{H^*}_{v(H^*)}
\qquad\text{and}\qquad
s(p',r',H^*)=\Delta'.
\]

Suppose for contradiction that no such component exists.

For each $H\in D$, let $W(H)\in\mathcal C$ be the witness given by
\Cref{prop:component-witness}; thus
\[
p_j^{W(H)}=p'_j \quad \forall j\in H\cap G,
\qquad
r_i^{W(H)}=r'_i \quad \forall i\in H\cap B.
\]
If $W(H)=H$ for some $H\in D$, then
\[
p'_j=p_j^H \quad \forall j\in H\cap G,
\qquad
r'_i=r_i^H \quad \forall i\in H\cap B,
\]
and therefore
\[
s(p',r',H)=s(p^H,r^H,H)=\Delta',
\]
contrary to our assumption. Hence
\[
W(H)\neq H
\qquad\forall H\in D.
\]

Moreover, for every $H\in D$, the run rooted at $W(H)$ must actually change at
least one coordinate on $H$, because otherwise $(p^{W(H)},r^{W(H)})=(p,r)$ on
$H$ and the witness would coincide with the $H$-run. Therefore
\[
s(p,r,W(H))>\Delta',
\]
so \Cref{alg:getprices-aa} invokes \textsc{SpecialPriceAA} on $W(H)$, and by
part~(2),
\[
s(p^{W(H)},r^{W(H)},W(H))=\Delta'.
\]
Thus
\[
W(H)\in D
\qquad\forall H\in D.
\]

Since $D$ is finite and $W$ maps $D$ into itself without fixed points, there
exists a directed cycle
\[
K_1,K_2,\dots,K_t,K_{t+1}=K_1
\]
in $D$ such that
\[
W(K_c)=K_{c+1}
\qquad\text{for } c=1,\dots,t.
\]

Fix $c\in\{1,\dots,t\}$. Since $K_{c+1}$ is the witness for $K_c$, we have
\[
p_j^{K_{c+1}}=p'_j\ge p_j^{K_c}
\qquad\forall j\in K_c\cap G,
\]
and
\[
r_i^{K_{c+1}}=r'_i\ge r_i^{K_c}
\qquad\forall i\in K_c\cap B,
\]
with at least one strict inequality on the component $K_c$, because
$W(K_c)\neq K_c$.

Since all goods of a fixed component are scaled uniformly in any run, exactly
one of the following two cases holds for each $c$:
\begin{enumerate}
    \item[(a)]
    \[
    p_{v(K_c)}^{K_{c+1}} > p_{v(K_c)}^{K_c};
    \]
    \item[(b)]
    \[
    p_{v(K_c)}^{K_{c+1}} = p_{v(K_c)}^{K_c},
    \]
    and then some refund must be strictly larger on $K_c$.
\end{enumerate}

By the perturbation assumption from \Cref{subsec:generic-perturbation}, each
component contains at most one critical buyer. Hence in case~(b) there is a
unique buyer $i_c\in K_c\cap B$ such that
\[
r_{i_c}^{K_{c+1}} > r_{i_c}^{K_c}.
\]

We now distinguish two possibilities.

\smallskip
\noindent\textbf{Case 1:} Case~\textup{(a)} holds for at least one index $c$.

For every $c$, either case~(a) holds, or case~(b) holds. In both situations,
the component $K_c$ is active in the run rooted at $K_{c+1}$: in case~(a) this
follows from \Cref{lem:domination}; in case~(b), the buyer $i_c$ receives
additional refund in the $K_{c+1}$-run, so $i_c$ becomes critical there, hence
$K_c\subseteq \operatorname{ActiveSet}(p^{K_{c+1}},\Delta,K_{c+1})$.

Therefore \Cref{lem:domination} gives
\[
\mu(K_{c+1},K_c)
=
\frac{p_{v(K_c)}^{K_{c+1}}}{p_{v(K_{c+1})}^{K_{c+1}}}
\qquad\text{for every }c.
\]
Concatenating the corresponding maximum-utility paths, we obtain a closed walk
$W^*$ in the auxiliary network $N^*$ with
\[
U(W^*)=
\prod_{c=1}^t \mu(K_{c+1},K_c)
=
\prod_{c=1}^t
\frac{p_{v(K_c)}^{K_{c+1}}}{p_{v(K_{c+1})}^{K_{c+1}}}.
\]
Since
\[
p_{v(K_c)}^{K_{c+1}}\ge p_{v(K_c)}^{K_c}
\qquad\text{for every }c,
\]
and the inequality is strict for at least one index by assumption of Case~1,
we get
\[
U(W^*)>
\prod_{c=1}^t
\frac{p_{v(K_c)}^{K_c}}{p_{v(K_{c+1})}^{K_{c+1}}}
=
\frac{\prod_{c=1}^t p_{v(K_c)}^{K_c}}
     {\prod_{c=1}^t p_{v(K_{c+1})}^{K_{c+1}}}
=1,
\]
where the last equality is a cyclic re-indexing. This contradicts
\Cref{lem:cycle-utility}.

\smallskip
\noindent\textbf{Case 2:} Case~\textup{(b)} holds for every index $c$.

Thus, for every $c$,
\[
p_{v(K_c)}^{K_{c+1}} = p_{v(K_c)}^{K_c}
\qquad\text{and}\qquad
r_{i_c}^{K_{c+1}} > r_{i_c}^{K_c},
\]
where $i_c$ is the unique buyer of $K_c$ receiving new refund.

We claim that, in the run rooted at $K_t$, all components
\[
K_1,K_2,\dots,K_{t-1}
\]
are active.

Indeed, $K_{t-1}$ is active in the $K_t$-run because the buyer $i_{t-1}$
receives additional refund there. Now assume inductively that $K_c$ is active
in the $K_t$-run for some $2\le c\le t-1$. Since
\[
p_{v(K_{c-1})}^{K_c}=p_{v(K_{c-1})}^{K_{c-1}}
\qquad\text{and}\qquad
r_{i_{c-1}}^{K_c}>r_{i_{c-1}}^{K_{c-1}},
\]
the buyer $i_{c-1}$ becomes critical during the $K_c$-run at the same terminal
scale on the component $K_{c-1}$ as in the run rooted at $K_{c-1}$. Because
the $K_t$-run already activates $K_c$ and reaches on $K_c$ the same terminal
scale as the $K_c$-run, it reproduces this activation of $K_{c-1}$ as well.
Hence $K_{c-1}$ is active in the $K_t$-run. By backward induction, all
components $K_1,\dots,K_{t-1}$ are active in the run rooted at $K_t$.

Since every abundant edge remains an equality edge throughout
\textsc{SpecialPriceAA}, the whole active set in the $K_t$-run lies inside one
connected component of the equality graph at the terminal price vector
$p^{K_t}$. But in that same run, each buyer
\[
i_1,i_2,\dots,i_t
\]
becomes critical, because each of them receives additional refund along the
cycle. Thus one connected component of the equality graph contains at least $t$
distinct critical buyers, contradicting the perturbation assumption from
\Cref{subsec:generic-perturbation}, which states that every connected
component of every equality graph contains at most one critical buyer.

Both cases lead to a contradiction. Therefore our original assumption was
false, and there exists a non-singleton component $H^*\in D$ such that
\[
p'_{v(H^*)}=p^{H^*}_{v(H^*)}
\qquad\text{and}\qquad
s(p',r',H^*)=\Delta'.
\]
This proves part~(3).

\medskip
\noindent\emph{Part~(4).}
If part~(3)(b) holds, then alternative~(b) is immediate.

Thus assume part~(3)(a) holds, and let $H^*\in\mathcal C$ be the
non-singleton component given there, so that
\[
p'_{v(H^*)}=p^{H^*}_{v(H^*)}
\qquad\text{and}\qquad
s(p',r',H^*)=\Delta'.
\]
Since $H^*\in D$, part~(2) also gives
\[
s(p^{H^*},r^{H^*},H^*)=\Delta'.
\]

Suppose for contradiction that every component $J\in\mathcal C$ satisfies
\[
s(p',r',J)>-\Delta'/(2n^2).
\]
Because $p'\ge p^{H^*}$ coordinatewise and $r'\ge r^{H^*}$ coordinatewise, we
have
\[
s(p^{H^*},r^{H^*},J)\ge s(p',r',J)>-\Delta'/(2n^2)
\qquad\text{for every }J\in\mathcal C.
\]
Therefore, at the stopping state $(p^{H^*},r^{H^*})$ of the target-$\Delta'$
run rooted at $H^*$, the root component still has surplus
\[
s(p^{H^*},r^{H^*},H^*)=\Delta'>0,
\]
and every other component satisfies
\[
s(p^{H^*},r^{H^*},J)>-\frac{\Delta'}{2n^2}
=
-\frac{s(p^{H^*},r^{H^*},H^*)}{2n^2}.
\]
Hence the target-$0$ run
\[
\textsc{SpecialPriceAA}(p,x,r,H^*,0)
\]
does not stop at the state $(p^{H^*},r^{H^*})$ and must continue strictly past
it.

During \textsc{SpecialPriceAA}, the surplus of the root component is
nonincreasing, and whenever the procedure continues past a state with positive
root surplus, the subsequent price increase on the active goods of $H^*$ makes
that surplus strictly smaller. Consequently, the target-$0$ run rooted at
$H^*$ ends with
\[
\Delta_{H^*}<\Delta'.
\]
But $H^*\in D$, so by definition of $D$ we have
\[
\Delta_{H^*}=\Delta',
\]
a contradiction.

Therefore there exists a component $J\in\mathcal C$ such that
\[
-\Delta'/n^2\le s(p',r',J)\le -\Delta'/(2n^2).
\]
This proves part~(4)(a).
\end{proof}
\subsection{Allocation reconstruction}
\label{subsec:getallocations}

Having produced $\Delta'$ and $(p',r')$, we now rebuild an allocation whose support lies entirely in the abundant graph $\mathcal A(\Delta)$.
For each $\Delta$-component $H$ with $|H|\ge 2$ we fix a \emph{buyer root} $i(H)\in H\cap B$ and a \emph{good root} $j(H)\in H\cap G$ (one of each; at least one of the two exists and if only one side is present we reuse the same node for both).
If $H=\{k\}$ is a singleton, we set $i(H)=j(H)=k$.

\begin{algorithm}[H]
\caption{\textsc{GetAllocationsAA}$(p',r',\Delta')$}
\label{alg:getallocations-aa}
\begin{algorithmic}[1]
\State Initialize $x'_{ij}\gets 0$ for every $(i,j)\in B\times G$
\For{each non-singleton component $H\in\mathcal C$}
    \State Select a buyer root $i(H)\in H\cap B$ and a good root $j(H)\in H\cap G$
    \For{each buyer $i\in H\cap B$}
        \If{$i=i(H)$}
            \State $\bar c_i(x') \gets \max\{0,\, s(p',r',H)\}$
        \Else
            \State $\bar c_i(x') \gets 0$
        \EndIf
    \EndFor
    \For{each good $j\in H\cap G$}
        \If{$j=j(H)$}
            \State $b_j(p',x') \gets \min\{0,\, s(p',r',H)\}$
        \Else
            \State $b_j(p',x') \gets 0$
        \EndIf
    \EndFor
    \State Compute the unique tree flow $x'$ on $A(\Delta)\cap (H\times H)$
    satisfying the above buyer-cash and good-backorder equations
\EndFor
\State \Return $x'$
\end{algorithmic}
\end{algorithm}
The supplies and demands are balanced within each non-singleton component $H$: summing,
\[
\sum_{i\in H\cap B}s_i-\sum_{j\in H\cap G}d_j=\bar e'_{H\cap B}-\max\{0,\tau_H\}-p'_{H\cap G}-\min\{0,\tau_H\}=\tau_H-\max\{0,\tau_H\}-\min\{0,\tau_H\}=0,
\]
so the tree flow $x'$ is well defined and unique.

\begin{lemma}[Tree-flow stability]
\label{lem:tree-flow-stability}
Let $T$ be a tree on bipartition $X\cup Y$.
Let $a,a'\in\mathbb R_{\ge 0}^{X}$ and $d,d'\in\mathbb R_{\ge 0}^{Y}$ satisfy
\[
\sum_{x\in X} a_x=\sum_{y\in Y} d_y,
\qquad
\sum_{x\in X} a_x'=\sum_{y\in Y} d_y'.
\]
Let $f$ and $f'$ be the unique flows on $T$ with buyer-side supplies
$a,a'$ and good-side demands $d,d'$, respectively.
Then for every edge $e\in T$,
\[
|f_e-f_e'|
\le
\frac12\left(
\sum_{x\in X}|a_x-a_x'|+\sum_{y\in Y}|d_y-d_y'|
\right).
\]
\end{lemma}

\begin{proof}
Delete the edge $e$ from $T$, and let $C$ be one of the two connected components
of $T-e$.
Because $T$ is a tree, the value on $e$ is uniquely determined by flow
conservation:
\[
f_e=\sigma_e\!\left(\sum_{x\in C\cap X} a_x-\sum_{y\in C\cap Y} d_y\right),
\]
where $\sigma_e\in\{\pm 1\}$ depends only on the orientation chosen for $e$.
The same formula holds for $f_e'$ with $(a',d')$.
Subtracting gives
\[
|f_e-f_e'|
\le
\sum_{x\in C\cap X}|a_x-a_x'|+\sum_{y\in C\cap Y}|d_y-d_y'|.
\]
Applying the same argument to the complementary component and using the fact that
the total imbalance over all nodes is zero for both pairs $(a,d)$ and $(a',d')$,
we obtain the identical bound on the complement.
Therefore $|f_e-f_e'|$ is bounded by the smaller of the two side-variations, which
is at most half of the full variation.
\end{proof}

\begin{lemma}[Output of \textsc{GetAllocationsAA}]
\label{lem:getallocations-output}
Assume $\Delta'\le \Delta/n^2$ and that the restart invariants of
\Cref{lem:restart-invariants} hold. Then the allocation $x'$ produced by
\Cref{alg:getallocations-aa} satisfies:
\begin{enumerate}
    \item $\bar c_i(x')\ge 0$ for every $i\in B$;
    \item $-\Delta'/n^2 \le b_j(p',x')\le 0$ for every $j\in G$;
    \item if $x'_{ij}>0$, then $(i,j)\in A(\Delta)\subseteq E(p')$;
    \item every previously $\Delta$-abundant edge $(i,j)$ satisfies
    \[
    x'_{ij}>3n\Delta'.
    \]
\end{enumerate}
\end{lemma}

\begin{proof}
Parts (1) and (2) are immediate from the construction of
\Cref{alg:getallocations-aa}. On each non-singleton component $H$, only the
buyer root $i(H)$ and the good root $j(H)$ absorb the component surplus
$s(p',r',H)$. Therefore every buyer has nonnegative effective cash, and every
good has backorder in the interval $[-\Delta'/n^2,0]$ by
\Cref{lem:restart-invariants}(1).

Part (3) is also immediate: \Cref{alg:getallocations-aa} routes flow only on
edges of $A(\Delta)$. Since every abundant edge remains an equality edge
throughout every execution of \textsc{SpecialPriceAA}, and $(p',r')$ is assembled from those
runs, it follows that
\[
A(\Delta)\subseteq E(p').
\]

We now prove part (4) .For each non-singleton component $H\in\mathcal C$, let $\hat x$ be the unique
tree flow on $A(\Delta)\cap(H\times H)$ such that
\begin{enumerate}
    \item $\hat x_{ij}=0$ for every $(i,j)\notin A(\Delta)$;
    \item $\bar c(\hat x)\ge 0$, and if $s(p,r,H)\ge 0$, then
    \[
    \bar c_H(\hat x)=s(p,r,H);
    \]
    \item $b(p,\hat x)\le 0$, and if $s(p,r,H)<0$, then
    \[
    b_H(p,\hat x)=s(p,r,H).
    \]
\end{enumerate}

We first claim that
\[
\bar c_B(\hat x)<(n+1)\Delta.
\]
Indeed,
\[
\bar c_B(\hat x)+b_G(p,\hat x)=\bar e_B-p_G
                           =\bar c_B(x)+b_G(p,x).
\]
Since $(p,x,r)$ is $\Delta$-optimal, we have $\bar c_i(x)<\Delta$ for every
buyer, and since $(p,x,r)$ is $\Delta$-feasible, every good with
$p_j>p_j^0$ has $b_j(p,x)\le \Delta$. Therefore
\[
\bar c_B(x)+b_G(p,x)\le n\Delta.
\]
On the other hand, because the state is not $\Delta$-fertile, every component
$H\in\mathcal C$ satisfies
\[
s(p,r,H)>-\Delta/(3n^2).
\]
Hence
\[
b_G(p,\hat x)
=
\sum_{H\in\mathcal C:\, s(p,r,H)<0} s(p,r,H)
>
-|\mathcal C|\frac{\Delta}{3n^2}
\ge -\Delta.
\]
Combining the two displays gives
\[
\bar c_B(\hat x)<(n+1)\Delta.
\]

Now consider the flow decomposition of
\[
y:=x-\hat x
\]
in the $\Delta$-residual network $N_\Delta(p)$. Exactly as in Orlin’s proof,
the total flow emanating from excess nodes is bounded by $\bar c_B(\hat x)$.
Therefore, for every abundant edge $(i,j)\in A(\Delta)$,
\[
|x_{ij}-\hat x_{ij}|<(n+1)\Delta.
\]

Next compare $\hat x$ with the new allocation $x'$ returned by
\Cref{alg:getallocations-aa}. On each component $H$, both $\hat x$ and $x'$
are tree flows on the same tree $A(\Delta)\cap(H\times H)$, and the only
difference between their node-balance data is at the two roots of $H$. Writing
\[
\tau_H:=s(p,r,H),
\qquad
\tau_H':=s(p',r',H),
\]
the total absolute variation of the buyer-root / good-root data on $H$ is
exactly
\[
|\tau_H-\tau_H'|.
\]
Hence, by \Cref{lem:tree-flow-stability}, every abundant edge
$(i,j)\in A(\Delta)\cap(H\times H)$ satisfies
\[
|\hat x_{ij}-x'_{ij}|
\le \frac12 |\tau_H-\tau_H'|.
\]

We now bound the global variation of the component surpluses. Since
$p'\ge p$ and $r'\ge r$, we have $\tau_H'\le \tau_H$ for every $H$, and thus
\[
\sum_{H\in\mathcal C} |\tau_H-\tau_H'|
=
\sum_{H\in\mathcal C} (\tau_H-\tau_H')
=
(\bar e_B-p_G)-(\bar e_B'-p_G').
\]
The first term equals
\[
\bar e_B-p_G=\bar c_B(x)+b_G(p,x)\le n\Delta,
\]
as above. For the second term, using
\Cref{lem:restart-invariants}(1) and $|\mathcal C|\le n$, we obtain
\[
\bar e_B'-p_G'
=
\sum_{H\in\mathcal C}s(p',r',H)
\ge
-|\mathcal C|\frac{\Delta'}{n^2}
\ge -\Delta.
\]
Therefore
\[
\sum_{H\in\mathcal C} |\tau_H-\tau_H'| < (n+1)\Delta.
\]
In particular, for every abundant edge,
\[
|\hat x_{ij}-x'_{ij}|<n\Delta.
\]

Combining the last two bounds gives, for every previously $\Delta$-abundant
edge $(i,j)$,
\[
|x_{ij}-x'_{ij}|
\le |x_{ij}-\hat x_{ij}|+|\hat x_{ij}-x'_{ij}|
< (n+1)\Delta+n\Delta=(2n+1)\Delta.
\]
Since $(i,j)$ was $\Delta$-abundant,
\[
x_{ij}\ge 3n\Delta,
\]
and therefore
\[
x'_{ij}>3n\Delta-(2n+1)\Delta=(n-1)\Delta.
\]
Because $\Delta'\le \Delta/n^2$, we have
\[
(n-1)\Delta \ge 3n\Delta'
\qquad (n\ge 3),
\]
and thus
\[
x'_{ij}>3n\Delta'.
\]
This proves part (4).
\end{proof}
\subsection{Proof of the restart lemma}
\label{subsec:restart-proof}

We now prove the restart lemma, \Cref{thm:restart}. Throughout this subsection,
$(p,x,r)$ denotes a $\Delta$-optimal compressed Arctic state that is not
$\Delta$-fertile. We write
\[
    \Delta' := \textsc{GetParameterAA}(p,x,r,\Delta),
    \qquad
    \mathcal C := \mathcal C(\Delta).
\]
The restart subroutine has two branches. If $\Delta'>\Delta/n^2$, then it does not
change the current state and only lowers the threshold. If $\Delta'\le \Delta/n^2$,
then it performs a compressed restart at scale $\Delta'$. In both cases we prove
that the algorithm makes progress within $O(\log n)$ ordinary scaling phases.

The threshold used by \textsc{MakeFertileAA} is
\[
    \mathrm{Threshold}\leftarrow \frac{\Delta}{n^5}.
\]
We use the following notion of progress. A progress event is either:
\begin{enumerate}
    \item the appearance of a new abundant edge; or
    \item a singleton buyer component $\{i\}$ reaching $\alpha_i(p)\le 1$.
\end{enumerate}
\begin{lemma}[Delayed-discovery progress bound]
\label{lem:delayed-discovery-progress}
Suppose that
\[
    \Delta'>\frac{\Delta}{n^2}.
\]
Then, after \textsc{MakeFertileAA} returns the delayed-discovery branch and sets
\[
    \mathrm{Threshold}\leftarrow \frac{\Delta}{n^5},
\]
the main algorithm discovers a new abundant edge within $O(\log n)$ ordinary scaling
phases. Moreover, \textsc{MakeFertileAA} is not re-invoked before this discovery.
\end{lemma}

\begin{proof}
Because the state is not $\Delta$-fertile, no singleton buyer component can be
responsible for the value $\Delta'>\Delta/n^2$. Indeed, if $H=\{i\}$ and
$\alpha_i(p)\le 1$, then \Cref{alg:getparameter-aa} sets $\Delta_H=0$.
If $H=\{i\}$ and $\alpha_i(p)>1$, then non-fertility gives
\[
    \Delta_H=\bar c_i(x)\le \frac{\Delta}{3n^2}<\frac{\Delta}{n^2}.
\]
Hence there exists a non-singleton component $H\in\mathcal C$ such that
\[
    \Delta_H=\Delta'.
\]

Let
\[
    (p^H,r^H):=\textsc{SpecialPriceAA}(p,x,r,H,0).
\]
By the definition of $\Delta_H$,
\[
    s(p^H,r^H,H)=\Delta'>\frac{\Delta}{n^2}>0.
\]
Thus the target-zero run did not stop because the surplus of $H$ reached zero.
By \Cref{lem:specialprice-output}, it must have stopped because some component
$J\in\mathcal C$ reached the lower barrier:
\[
    s(p^H,r^H,J)
    =
    -\frac{s(p^H,r^H,H)}{2n^2}
    =
    -\frac{\Delta'}{2n^2}.
\]
In particular,
\[
    s(p^H,r^H,J)<-\frac{\Delta}{2n^4}.
\]

We now run the ordinary scaling algorithm from $(p,x,r)$. We show that before
\textsc{MakeFertileAA} can be invoked again, either a new abundant edge has already
appeared, or the current state has become fertile. Since the new threshold is
$\Delta/n^5$, it is enough to consider the first later scale parameter
$\Delta''$ satisfying
\[
    \Delta''<\frac{\Delta}{n^5}.
\]
Let $(\hat p,\hat x,\hat r)$ be the state at the end of the $\Delta''$-phase, before
the next halving repair. Thus $(\hat p,\hat x,\hat r)$ is $\Delta''$-optimal.

Assume, for contradiction, that no new abundant edge has appeared before or during
this $\Delta''$-phase, and that the state $(\hat p,\hat x,\hat r)$ is not
$\Delta''$-fertile.

Since no new abundant edge has appeared, the old abundant graph is still
$A(\Delta)$, and the old components are still the components in $\mathcal C(\Delta)$.
Moreover, every old abundant edge remains positive, and every positive-spending edge
is an equality edge. Therefore every old component moves as a single price block:
for every component $K\in\mathcal C$ and every two goods $j,j'\in K\cap G$,
\[
    \frac{\hat p_j}{p_j}
    =
    \frac{\hat p_{j'}}{p_{j'}}.
\]

Fix the root good $v(H)\in H\cap G$. Define the multiplier reached on $H$ by the
special-price run as
\[
    \lambda_H:=\frac{p^H_{v(H)}}{p_{v(H)}}.
\]
At the ordinary state, define
\[
    \hat\lambda_H:=\frac{\hat p_{v(H)}}{p_{v(H)}}.
\]

We distinguish two cases.

\medskip
\noindent
\emph{Case 1: $\hat\lambda_H\le \lambda_H$.}
In this case the ordinary trajectory has not raised the price block of $H$ beyond
the level reached by the special-price run. During
\textsc{SpecialPriceAA}$(p,x,r,H,0)$, the surplus of $H$ can decrease only through
price increases on goods in the active set of $H$ and through refund commitments at
buyers in the active set of $H$. The ordinary trajectory is monotone in prices and
committed refunds. Hence, before the multiplier on $H$ exceeds $\lambda_H$, the
surplus of $H$ cannot be smaller than the surplus reached by the special-price run:
\[
    s(\hat p,\hat r,H)\ge s(p^H,r^H,H)=\Delta'.
\]

We convert this positive surplus into cut flow. For a component $S\in\mathcal C$,
define
\[
    X^+_S
    :=
    \sum_{i\in S\cap B}\sum_{j\notin S\cap G}\hat x_{ij},
    \qquad
    X^-_S
    :=
    \sum_{i\notin S\cap B}\sum_{j\in S\cap G}\hat x_{ij}.
\]
Here $X^+_S$ is the spending leaving $S$, and $X^-_S$ is the spending entering $S$.
Using
\[
    \bar c_i(\hat x)=\bar e_i-\sum_{j\in G}\hat x_{ij},
    \qquad
    b_j(\hat p,\hat x)=-\hat p_j+\sum_{i\in B}\hat x_{ij},
\]
we obtain the component identity
\[
    X^+_S-X^-_S
    =
    s(\hat p,\hat r,S)
    -
    \sum_{i\in S\cap B}\bar c_i(\hat x)
    -
    \sum_{j\in S\cap G} b_j(\hat p,\hat x).
\]

Because $(\hat p,\hat x,\hat r)$ is $\Delta''$-optimal, we have
\[
    \sum_{i\in B}\bar c_i(\hat x)<n\Delta''.
\]
Also, by $\Delta''$-feasibility, each good has
\[
    0\le b_j(\hat p,\hat x)\le \Delta''.
\]
Therefore, for every component $S$,
\[
    \sum_{i\in S\cap B}\bar c_i(\hat x)
    +
    \sum_{j\in S\cap G} b_j(\hat p,\hat x)
    \le 2n\Delta''.
\]
Applying this to $H$, we get
\[
    X^+_H
    \ge
    s(\hat p,\hat r,H)-2n\Delta''
    \ge
    \Delta'-2n\Delta''
    >
    \frac{\Delta}{n^2}-2n\Delta''.
\]
All positive-spending edges lie in the equality graph, and the equality graph is
cycle-free by the perturbation assumption. Hence at most $n-1$ positive cut edges
leave $H$. Thus some cut edge leaving $H$ carries more than
\[
    \frac{1}{n}
    \left(
        \frac{\Delta}{n^2}-2n\Delta''
    \right)
    =
    \frac{\Delta}{n^3}-2\Delta''.
\]
Since $\Delta''<\Delta/n^5$, this quantity is larger than $3n\Delta''$ whenever
\[
    n^2>3n+2.
\]
This holds for all $n\ge 4$. The remaining cases $n<4$ are easy and can be absorbed
into the constant-size base case of the algorithm. Hence some edge crossing the cut
of $H$ carries more than $3n\Delta''$.

This edge crosses two distinct old $\Delta$-components, so it is not in $A(\Delta)$.
Therefore it is a new $\Delta''$-abundant edge, contradicting our assumption.

\medskip
\noindent
\emph{Case 2: $\hat\lambda_H>\lambda_H$.}
In this case the ordinary trajectory has raised the price block of $H$ beyond the
level reached by the special-price run. We show that the component $J$ that became
tight in the special-price run must have reached at least the same deficit in the
ordinary trajectory.

Since $J$ became tight during the run rooted at $H$, it belongs to the active set of
$H$ at the stopping state $(p^H,r^H)$. Therefore
\Cref{lem:domination} gives
\[
    \mu(H,J)
    =
    \frac{p^H_{v(J)}}{p^H_{v(H)}}.
\]
Now apply the same path-product upper bound from the proof of
\Cref{lem:domination} to the ordinary price vector $\hat p$. Since no new
abundant edge has appeared, the auxiliary network is still built from the same old
abundant graph $A(\Delta)$, and every old abundant edge remains an equality edge.
Thus
\[
    \mu(H,J)
    \le
    \frac{\hat p_{v(J)}}{\hat p_{v(H)}}.
\]
Combining the last two displays gives
\[
    \frac{\hat p_{v(J)}}{\hat p_{v(H)}}
    \ge
    \frac{p^H_{v(J)}}{p^H_{v(H)}}.
\]
Since $\hat p_{v(H)}>p^H_{v(H)}$, we obtain
\[
    \hat p_{v(J)}>p^H_{v(J)}.
\]
This is the domination phenomenon formalized in \Cref{lem:domination}: once the
price block of a component is driven beyond the level reached in the corresponding
special-price run, every component connected to it by the maximum-multiplier path is
reached as well. In particular, the ordinary trajectory has raised the price block of
$J$ at least as far as the special-price run did.

Refund commitments are monotone. Moreover, a refund commitment in
\textsc{SpecialPriceAA} occurs only when the corresponding buyer becomes critical.
Once a buyer becomes critical, she remains
weakly critical at all later larger price vectors. Hence the ordinary trajectory has
also made, or can safely make, all refund commitments inside $J$ that were made by
the special-price run. These commitments only decrease the surplus of $J$. Therefore
\[
    s(\hat p,\hat r,J)
    \le
    s(p^H,r^H,J)
    =
    -\frac{\Delta'}{2n^2}
    <
    -\frac{\Delta}{2n^4}.
\]
Since $\Delta''<\Delta/n^5$, we have
\[
    \frac{\Delta}{2n^4}
    >
    \frac{\Delta''}{3n^2}.
\]
Thus
\[
    s(\hat p,\hat r,J)
    <
    -\frac{\Delta''}{3n^2}.
\]
So $J$ is $\Delta''$-fertile by condition~\textup{(2)} of
\Cref{def:delta-fertile-strong}, contradicting our assumption that
$(\hat p,\hat x,\hat r)$ is not $\Delta''$-fertile.

\medskip
Both cases lead to a contradiction. Therefore, before \textsc{MakeFertileAA} can be
invoked again, either a new abundant edge has appeared, or the state has become
fertile.

If a new abundant edge has appeared, we are done. Otherwise, the state is fertile at
some scale below the threshold. Since \textsc{MakeFertileAA} is called only when the
state is non-fertile, it is not re-invoked. Applying
\Cref{lem:fertile-gives-abundant} to the fertile component gives a new abundant edge
within at most $5\log_2 n+5$ additional ordinary phases. The number of ordinary
halvings before the scale first falls below $\Delta/n^5$ is at most
$5\log_2 n+1$. Hence a new abundant edge appears within $O(\log n)$ ordinary phases,
and \textsc{MakeFertileAA} is not re-invoked before that happens.
\end{proof}

\begin{lemma}[Compressed-restart progress and feasibility]
\label{lem:compressed-restart-progress}
Suppose that
\[
    \Delta'\le \frac{\Delta}{n^2}.
\]
Let
\[
    (p',r'):=\textsc{GetPricesAA}(p,x,r,\Delta'),
    \qquad
    x':=\textsc{GetAllocationsAA}(p',r',\Delta')
\]
be the state produced by the compressed-restart branch of
\textsc{MakeFertileAA}. Then:
\begin{enumerate}
    \item Every previously $\Delta$-abundant edge remains $\Delta'$-abundant.
    \item The restarted state can be used as a relaxed $\Delta'$-feasible state, and
    the ordinary scaling procedure restores ordinary $\Delta'$-feasibility without
    destroying any previously abundant edge.
    \item After ordinary $\Delta'$-feasibility has been restored, a progress event
    occurs within $O(\log n)$ further ordinary scaling phases.
\end{enumerate}
\end{lemma}

\begin{proof}
Part~(1) is exactly \Cref{lem:getallocations-output}\textup{(4)}: every previously
$\Delta$-abundant edge $(i,j)$ satisfies
\[
    x'_{ij}>3n\Delta'.
\]
Thus every old abundant edge is $\Delta'$-abundant after the compressed restart.

We next explain the feasibility issue in part~(2). By
\Cref{lem:getallocations-output}, the allocation $x'$ satisfies
\[
    \bar c_i(x')\ge 0
    \qquad\forall i\in B,
\]
and
\[
    -\frac{\Delta'}{n^2}\le b_j(p',x')\le 0
    \qquad\forall j\in G.
\]
Moreover, if $x'_{ij}>0$, then $(i,j)\in A(\Delta)\subseteq E(p')$. Thus the only
difference from ordinary $\Delta'$-feasibility is the following harmless integrality
and backorder relaxation.

First, some old abundant edges may carry values that are not exact multiples of
$\Delta'$. This does not affect the scaling algorithm. These old abundant edges
remain strictly $\Delta'$-abundant, by part~(1), and hence they are never bottleneck
edges for the discrete $\Delta'$-augmentations. All new nonzero allocations created
after the restart are created by ordinary $\Delta'$-augmentations, and therefore are
multiples of $\Delta'$.

Second, some goods may have small negative backorder, in the interval
$[-\Delta'/n^2,0]$. During the ordinary repair from the relaxed state, augmentations
to such goods increase their backorders in increments of $\Delta'$. Since the deficit
of any such good is smaller than $\Delta'$, one augmentation makes its backorder
nonnegative and still at most $\Delta'$. Therefore the ordinary scaling procedure
restores the usual condition
\[
    0\le b_j(p',x')\le \Delta'
\]
without endangering any previously abundant edge. Consequently, after this repair,
the maintained state is an ordinary $\Delta'$-feasible state, except for the harmless
fixed offsets on old abundant edges described above. Since old abundant edges are
never bottlenecks and all newly created nonzero allocations are multiples of
$\Delta'$, the subsequent scaling analysis is unchanged.

It remains to prove part~(3). By \Cref{lem:restart-invariants}\textup{(4)}, one of
the following two alternatives holds.

First, there exists a component $J\in\mathcal C$ such that
\[
    -\frac{\Delta'}{n^2}
    \le
    s(p',r',J)
    \le
    -\frac{\Delta'}{2n^2}.
\]
Then
\[
    s(p',r',J)
    \le
    -\frac{\Delta'}{3n^2}.
\]
Thus, as soon as ordinary $\Delta'$-feasibility has been restored, the component
$J$ is $\Delta'$-fertile by condition~\textup{(2)} of
\Cref{def:delta-fertile-strong}. By \Cref{lem:fertile-gives-abundant}, within at
most $5\log_2 n+5$ further ordinary phases a new abundant edge appears.

Second, there exists a singleton buyer component $H^*=\{i^*\}$ such that
\[
    \bar e_{i^*}=\Delta'.
\]
If $\alpha_{i^*}(p')\le 1$, then the buyer has already reached the second type of
progress event, and this event is permanent. If instead
$\alpha_{i^*}(p')>1$, then at the restart allocation this singleton buyer has
effective cash
\[
    \bar c_{i^*}(x')=\bar e_{i^*}=\Delta'
    >
    \frac{\Delta'}{3n^2}.
\]
Therefore $\{i^*\}$ is $\Delta'$-fertile by condition~\textup{(1)} of
\Cref{def:delta-fertile-strong}. Applying \Cref{lem:fertile-gives-abundant}, within
at most $5\log_2 n+5$ ordinary phases either a new abundant edge appears or
$i^*$ permanently reaches $\alpha_{i^*}\le 1$. In either case a progress event
occurs.

Thus, after the relaxed restarted state has been restored to ordinary
$\Delta'$-feasibility, a progress event occurs within $O(\log n)$ ordinary scaling
phases.
\end{proof}

\begin{proof}[Proof of \Cref{thm:restart}]
Let
\[
    \Delta' := \textsc{GetParameterAA}(p,x,r,\Delta).
\]

\medskip
\noindent
\emph{Delayed-discovery branch.}
Suppose first that
\[
    \Delta'>\frac{\Delta}{n^2}.
\]
Then \textsc{MakeFertileAA} leaves the state and scale unchanged and sets
\[
    \mathrm{Threshold}\leftarrow \frac{\Delta}{n^5}.
\]
By \Cref{lem:delayed-discovery-progress}, the subsequent run of the main algorithm
discovers a new abundant edge within $O(\log n)$ ordinary scaling phases, and
\textsc{MakeFertileAA} is not re-invoked before that discovery. This is the
delayed-discovery guarantee.

\medskip
\noindent
\emph{Compressed-restart branch.}
Now suppose that
\[
    \Delta'\le \frac{\Delta}{n^2}.
\]
Then \textsc{MakeFertileAA} computes
\[
    (p',r')=\textsc{GetPricesAA}(p,x,r,\Delta')
\]
and
\[
    x'=\textsc{GetAllocationsAA}(p',r',\Delta'),
\]
and returns the compressed state $(p',x',r')$ at scale $\Delta'$.

By \Cref{lem:compressed-restart-progress}, every previously abundant edge remains
abundant at the new scale, and the relaxed restarted state can be brought back to
ordinary $\Delta'$-feasibility without destroying any previously abundant edge.
After feasibility is restored, the same lemma gives a progress event within
$O(\log n)$ additional ordinary scaling phases. Since
\[
    \Delta'\le \frac{\Delta}{n^2},
\]
this is exactly the compressed-restart guarantee.

\medskip
\noindent
\emph{Running time.}
It remains to bound the running time of \textsc{MakeFertileAA}. The procedures
\textsc{GetParameterAA} and \textsc{GetPricesAA} call
\textsc{SpecialPriceAA} at most once per component. Since there are at most $n$
components, the total number of calls to \textsc{SpecialPriceAA} made by these two
procedures is $O(n)$.

By \Cref{lem:specialprice-terminates}, each call to \textsc{SpecialPriceAA} performs
at most $n+|B|=O(n)$ iterations of its outer while-loop. Each iteration requires
an active-set update and a next-event computation, which can be implemented in
$O(m+n\log n)$ time using the same data structures as in the ordinary
price-and-augment step. Hence the total time spent in all calls to
\textsc{SpecialPriceAA} during one invocation of \textsc{MakeFertileAA} is
\[
    O\bigl(n^2(m+n\log n)\bigr).
\]
Finally, \textsc{GetAllocationsAA} solves tree-flow problems on the abundant graph,
which takes $O(m+n)$ total time over all components. This is dominated by the
previous bound.

Therefore one call to \textsc{MakeFertileAA} takes
\[
    O\bigl(n^2(m+n\log n)\bigr)
\]
time. This completes the proof of \Cref{thm:restart}.
\end{proof}